\documentclass[journal]{IEEEtran}

\usepackage{graphicx}          
\usepackage[dvips]{epsfig}    

\usepackage{color,soul}
                               
\usepackage{amssymb}
\usepackage{amsmath}
\usepackage{amsthm}
\usepackage{mathtools}
\DeclareMathOperator*{\argmax}{arg\,max}

\usepackage{algcompatible}
\usepackage{algorithmicx}
\usepackage[noend]{algpseudocode}
\usepackage{algorithm}
\algnewcommand\algorithmicforeach{\textbf{for each}}
\algdef{S}[FOR]{ForEach}[1]{\algorithmicforeach\ #1\ \algorithmicdo}

\newtheorem{proposition}{Proposition}
\newtheorem{Lemma}{Lemma}
\newtheorem{Theorem}{Theorem}
\newtheorem{Definition}{Definition}

\newtheorem{proof-theorem-2}{\textif{Proof of Theorem 2:}}

\usepackage{cite}
\usepackage{breqn}

\begin{document}


\title{LQG Reference Tracking with Safety and Reachability Guarantees under Unknown False Data Injection Attacks} 
                                                
\author{Zhouchi~Li,~\IEEEmembership{Student Member,~IEEE,}
        Luyao~Niu,~\IEEEmembership{Student Member,~IEEE,}
        and~Andrew~Clark,~\IEEEmembership{Member,~IEEE}
\thanks{Z. Li, L. Niu, and A. Clark are with the Department
of Electrical and Computer Engineering, Worcester Polytechnic Institute, Worcester,
MA, 01609 USA e-mail: \{zli4, lniu, aclark\}@wpi.edu.}}

\maketitle

\begin{abstract}                          
We investigate a linear quadratic Gaussian (LQG) tracking problem with safety and reachability constraints in the presence of an adversary who mounts an FDI attack on an unknown set of sensors. For each possible set of compromised sensors, we maintain a state estimator disregarding the sensors in that set, and calculate the optimal LQG control input at each time based on this estimate. 
We propose a control policy which constrains the control input to lie within a fixed distance of the optimal control input corresponding to each state estimate. The control input is obtained at each time step by solving a quadratically constrained quadratic program (QCQP). We prove that our policy can achieve a desired probability of safety and reachability using the barrier certificate method. Our control policy is evaluated via a numerical case study.

\end{abstract}


\section{Introduction}
Safety \cite{mitra2013verifying,banerjee2011ensuring} and reachability \cite{kwon2017reachability} are critical properties of control systems. The safety constraint requires that the system state should remain in a safe region. The reachability constraint requires that the system should reach a set of goal states within a desired time interval. Safety and reachability are fundamental requirements for critical applications including healthcare, transportation, and power systems.

Control systems have been shown to be vulnerable to malicious attacks. Various attacks targeting at actuators and measurement channels have been reported~\cite{liu2011false,teixeira2012attack}. Particularly, false data injection (FDI) attacks, which compromise the sensor measurements, need special concerns because they are easily mounted~\cite{liang20162015}, stealthy if the adversary knows the full information of the system~\cite{yu2015blind,mo2010false}, and can cause serious financial loss or even personal damage~\cite{liang2016review}.  
One example is GPS spoofing against unmanned aerial vehicles and autonomous cars, which results in deviation from the desired trajectory, as well as violations of safety and reachability \cite{kerns2014unmanned,petit2014potential}. The threat of such attacks has led to significant research interest in modeling~\cite{zhang2017game,hu2018state,zhang2020false}, mitigating~\cite{wu2018detecting}, and detecting FDI~\cite{mo2010false}\cite{liu2014detecting}. Resilient state estimation is also investigated in~\cite{shoukry2018smt,pajic2016attack,fawzi2014secure}. The authors of \cite{abdi2018preserving} aim at computing a safe operational windows to guarantee the safety property of a deterministic linear system with complete information. \cite{wu2018detecting} assumes that the correct sensor measurements of system state are always available to the controller even when the system is under attack. An emergency controller is assumed in \cite{gheitasi2019novel} which can be invoked when an alert on attacks is raised. In \cite{farivar2019artificial}, a single-input single-output system under false data injection attack targeting at actuator is studied.

At present, less attention has been paid to the design of closed-loop controllers with safety and reachability guarantees under FDI attacks. In the preliminary conference version of this work~\cite{niu2019lqg}, we investigated the linear quadratic Gaussian (LQG) reference tracking problem, in which there was only one possible set of compromised sensors. In this paper, we generalize the problem so that multiple possible compromised sensor sets are given, each of which corresponds to a different attack scenario. The goal of our approach is to develop a control policy that ensures safety and reachability under each attack scenario while also minimizing the LQG tracking cost when no adversary is present.  

Under our approach, for each attack pattern, the system maintains a state estimate that ignores the sensor measurements corresponding to that attack pattern. The control action chosen at each time step is then constrained to be within a fixed distance of the optimal control action corresponding to each state estimate. The key challenge is that, when there are multiple possible attack scenarios, the state estimates may be inconsistent from each other. To overcome this difficulty, we propose a scheme for detecting and resolving inconsistencies between state estimates. The selected state estimates are utilized to construct constraints that guarantee safety and reachability with desired probability.

The contribution of this paper is two-fold. First, a barrier certificate based policy is proposed to solve the LQG tracking problem with safety and reachability constraints under FDI attack that targets at an unknown set of sensors. We solve a quadratically constrained quadratic program (QCQP) to calculate the control policy at each time step. We develop a procedure to resolve the potential infeasibility of the QCQP. We prove that the controller obtained using our approach guarantees safety and reachability with desired probabilities. We show the feasibility and performance guarantees of the controller when the adversary is absent. Second, we derive a closed-form solution of the controller for a special case of the problem where there is a unique attack pattern. The derived controller not only guarantees safety and reachability, but also achieves better approximation with respect to the expected cost, compared with our preliminary work~\cite{niu2019lqg}.

The note is organized as follows. Section \ref{sec:system_model} states the system and adversary models and the problem formulation. Section \ref{Multiple-adversarial-scenario} and \ref{Single-adversarial-scenario} propose the control policy for the multiple- and single-adversary scenarios, respectively. Section \ref{sec:simulation} contains simulation results. Section \ref{sec:conclusion} concludes the paper.

\section{System Model and Problem Statement}
\label{sec:system_model}
In this section, we first present the system and adversary models. We then give the problem formulation. 

\subsection{System and Adversary Models}
We consider a linear time invariant (LTI) system with state $\mathbf{x}(t) \in \mathbb{R}^{n}$, input $\mathbf{u}(t) \in \mathbb{R}^{m}$, and observations $\mathbf{y}(t) = \left[y_1(t),\ldots,y_p(t)\right]^T \in \mathbb{R}^{p}$. The system dynamics are 
\begin{subequations}
\label{eq:dyn}
    \begin{align}
        \label{eq:dyn-state}
        \dot{\mathbf{x}}(t) &= A\mathbf{x}(t) + B\mathbf{u}(t) + \mathbf{w}(t) \\
        \label{eq:dyn-meas}
        \mathbf{y}(t) &= C\mathbf{x}(t) + \mathbf{v}(t) + \mathbf{a}(t)
    \end{align}
\end{subequations}
In Equation~(\ref{eq:dyn}), $\mathbf{w}(t)$ and $\mathbf{v}(t)$ are independent Gaussian processes with means identically zero and autocorrelation functions $R_{\mathbf{w}}(\tau) = \Sigma_{\mathbf{w}}\delta(\tau)$ and $R_{\mathbf{v}}(\tau) = \Sigma_{\mathbf{v}}\delta(\tau)$, respectively, where $\delta(\tau)$ denotes the Dirac delta function. 
We use $\Sigma_{\mathbf{w}}$ and $\Sigma_{\mathbf{v}}$ to denote the covariance matrices of $\mathbf{w}(t)$ and $\mathbf{v}(t)$ at each time $t$. We assume $(A, N_{\mathbf{w}})$ is stabilizable, where $N_{\mathbf{w}}N_{\mathbf{w}}^T = \Sigma_{\mathbf{w}}.$
The initial state $\mathbf{x}(0)$ is equal to $\mathbf{x}_{0}$.
Denote $\mathfrak{I}(t)$ as the information available to the controller at time $t$. We have $\mathfrak{I}(t) = \{y(t^\prime)|t^\prime\leq t\}\cup\{u(t^\prime)|t^\prime<t\}$ and $\mathfrak{I}(0) = y(0).$ The control policy of the system is defined as a function $\mu(\mathfrak{I}(t)) \in \mathbb{R}^{m}.$

In Equation~(\ref{eq:dyn-meas}), $\mathbf{a}(t)$ is the attack signal injected by the adversary. 
There exists a collection of attack patterns $\{\mathcal{A}_i: i \in \{1,\ldots,q\}\}.$ Here $\mathcal{A}_i \subseteq \{1,\ldots,p\}$ is a subset of sensors, in which it is possible that $\mathcal{A}_i \bigcap \mathcal{A}_j \neq \emptyset$. 
The adversary chooses one $\mathcal{A}^{\ast}$ from $\mathcal{A}_1,\ldots,\mathcal{A}_q.$ The adversary then chooses $\mathbf{a}(t)$ arbitrary values such that $support(\mathbf{a}(t)) \subseteq \mathcal{A}^{\ast}$ for all time $t \in [0,T].$ The controller knows the possible compromised sets $\mathcal{A}_1,\ldots,\mathcal{A}_q,$ but does not know which set $\mathcal{A}^{\ast}$ has been chosen by the adversary.
At each time $t,$ the adversary knows the control policy $\mu(\mathfrak{I}(t)),$ the system state $\mathbf{x}(t^{\prime}),$ the system output $\mathbf{y}(t^{\prime}),$ and the control input $\mathbf{u}(t^{\prime})$ for all $t^{\prime} \leq t.$ Denote the adversary policy $\nu(t)$ as a function which maps 
$\{\mathbf{x}(t^{\prime}), \mathbf{u}(t^{\prime}), \mathbf{y}(t^{\prime}): \forall t^{\prime} \leq t\}$ to $\mathbf{a}(t).$

Let $G$ and $U$ be the goal states and unsafe states defined as 
$G = \{\mathbf{x} \in \mathbb{R}^n: g_G(\mathbf{x}) \geq 0\},$ and $U = \{\mathbf{x} \in \mathbb{R}^n: g_U(\mathbf{x}) \geq 0\},$ respectively.
Define the safety constraint as $\mathbf{x}(t) \notin U \ \forall t \in [0,T],$ which prevents the system state from reaching $U$ for all time $t \in [0,T].$ 
We define the reachability constraint as
$\mathbf{x}(T) \in G,$ which requires the system state to be in $G$ at final time $T.$ A reference trajectory $\{\mathbf{r}(t) \in \mathbb{R}^n : t \in [0,T]\}$ is given, which satisfies $\mathbf{r}(t) \notin U$ and $\mathbf{r}(T) \in G$.

\subsection{Problem Formulation}
The problem studied in this work is 
\begin{align}
\label{eq:problem-form}
    \min_{\mu} \ & \mathbf{E}[(\mathbf{x}(T)-\mathbf{r}(T))^{T}F(\mathbf{x}(T)-\mathbf{r}(T)) + \int_{0}^{T}{(\mathbf{u}(t)^{T}R\mathbf{u}(t) }\nonumber\\
    &+ (\mathbf{x}(t)-\mathbf{r}(t))^{T}Q(\mathbf{x}(t)-\mathbf{r}(t))  
    ) \ dt | \mu, \beta = 1 ] \\
    \mbox{s.t.} \ & \max_{\nu}{\{Pr(\cup_{t \in [0,T]}{\{\mathbf{x}(t) \in U\}} | \mu, \nu)\}} \leq \epsilon_s  \nonumber\\ 
    & \min_{\nu}{\left\{Pr(\mathbf{x}(T) \in G | \mu, \nu)\right\}} \geq 1-\epsilon_r\nonumber
\end{align}
The objective function implies that the goal of the system is to minimize the expected cost when there is no adversary ($\beta = 1$), while guaranteeing safety and reachability when the system is under attack ($\beta = 0$).
The first constraint implies that the probability of violating the safety constraint in the worst case over all the adversary policies should be lower than the bound $\epsilon_s.$ The second constraint means that the probability of achieving the reachability constraint should be greater than the threshold $1 - \epsilon_r$ under any adversary's policy.

\section{Control Strategy for Multiple-adversary Scenario}
\label{Multiple-adversarial-scenario}

In this section, we present the solution approach for multi-adversary scenario. Our proposed control policy is illustrated in Figure~\ref{fig:block_diagram}. Our solution approach is based on the observation that the adversary can only bias the system state by injecting false measurements to the sensors to induce erroneous control inputs. Hence, if we can restrict the control inputs to stay within a particular neighborhood of each optimal control signal $\mathbf{u}_{\alpha,i}(t)$ that corresponds to the measurements from $\{1,\ldots,p\}\setminus \mathcal{A}_i$ for each attack pattern $\mathcal{A}_i$, then we can limit the impact from the adversary’s attack signal.

\begin{figure}[!ht]
\centering
\includegraphics[width=3.5in]{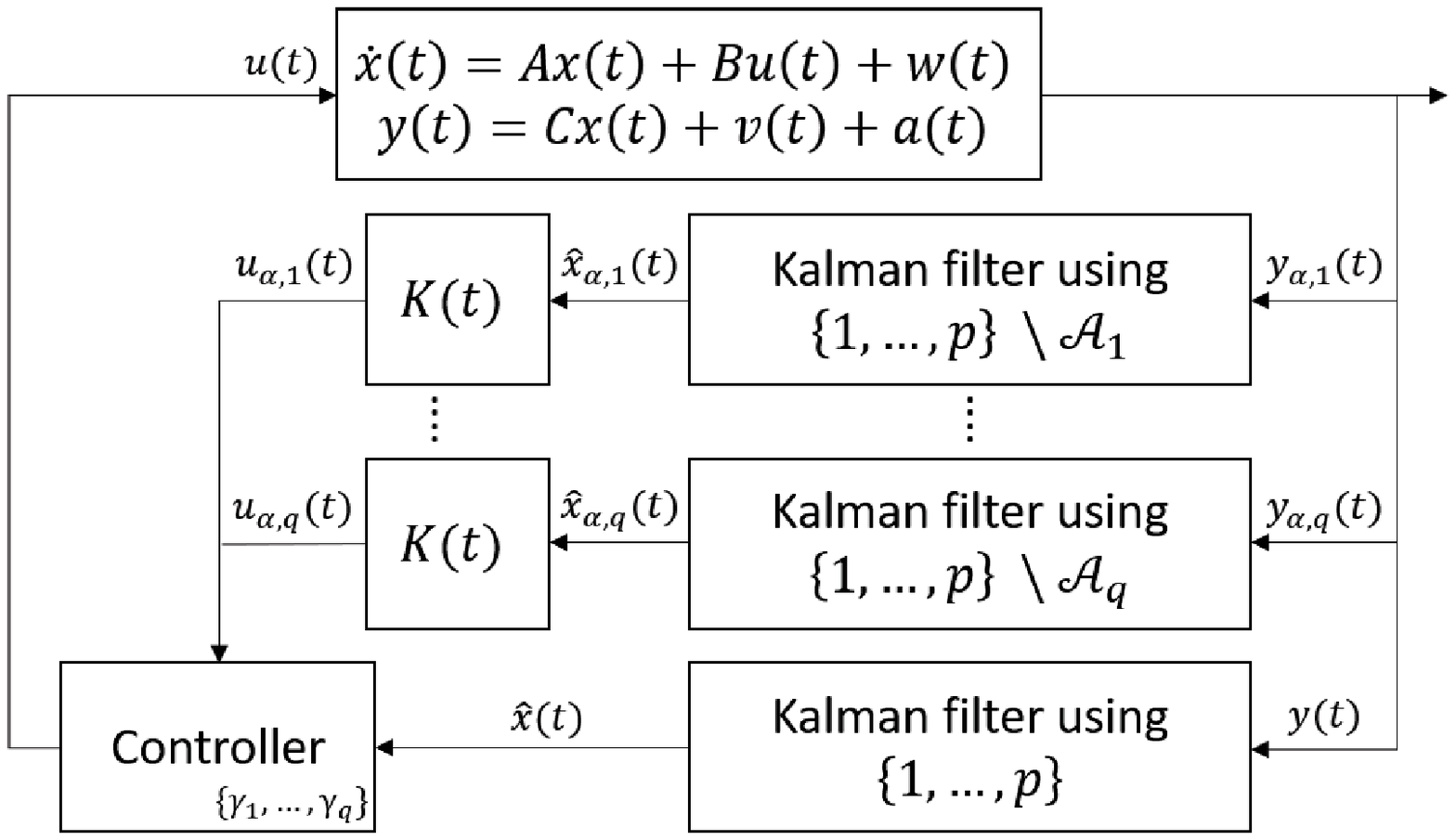}
\caption{Schematic illustration of our proposed approach.} 
\label{fig:block_diagram}
\end{figure}

\subsection{Control Policy Construction}
\label{subsec:control-policy}
 
Let $\mathbf{y}_{\alpha,i}(t)$ be the measurements of sensors in $\{1,\ldots,p\} \setminus \mathcal{A}_i$. Denote $C_{\alpha, i}$ and $\mathbf{v}_{\alpha, i}(t)$ as $C$ and $\mathbf{v}(t)$ with rows indexed by $\{1,\ldots,p\} \setminus \mathcal{A}_i,$ so that $\mathbf{y}_{\alpha, i}(t) = C_{\alpha, i}\mathbf{x}(t) + \mathbf{v}_{\alpha, i}(t)$. We assume that all systems $(A, C_{\alpha, i}) \ \forall i \in \{1,\ldots,q\}$ are observable. Let $\Sigma_{\mathbf{v}_{\alpha, i}}$ denote the covariance matrix of $\mathbf{v}_{\alpha, i}$. 

The Kalman Filter (KF) estimates $\hat{\mathbf{x}}(t)$ and $\hat{\mathbf{x}}_{\alpha,i}(t)$ are~\cite{anderson2007optimal}
\begin{equation*}
\label{eq:dynamics-x-hat}
    \begin{array}{ll}
        \dot{\hat{\mathbf{x}}}(t) &= A\hat{\mathbf{x}}(t) + B\mathbf{u}(t) + \Theta(t)(\mathbf{y}(t) - C\hat{\mathbf{x}}(t)) \\
        \Theta(t) &= \Phi(t)C^{T}\Sigma_{\mathbf{v}}^{-1} \\
        \dot{\Phi}(t) &= A\Phi(t) + \Phi(t)A^{T} + \Sigma_{\mathbf{w}} - \Phi(t)C^{T}\Sigma_{\mathbf{v}}^{-1}C\Phi(t)^{T}
    \end{array}
\end{equation*}
and
\begin{align}
    \label{eq:dynamics-x-hat1}
    \dot{\hat{\mathbf{x}}}_{\alpha,i}(t) =& A\hat{\mathbf{x}}_{\alpha,i}(t) + B\mathbf{u}_{\alpha,i}(t)\nonumber \\ 
    &+ \Theta_{\alpha,i}(t)(\mathbf{y}_{\alpha,i}(t) - C_{\alpha,i}\hat{\mathbf{x}}_{\alpha,i}(t))\\
    \label{eq:dynamics-x-hat2}
    \Theta_{\alpha,i}(t) =& \Phi_{\alpha,i}(t)C_{\alpha,i}^{T}\Sigma_{\mathbf{v}_{\alpha,i}}^{-1} \\
    \dot{\Phi}_{\alpha,i}(t) =& A\Phi_{\alpha,i}(t) + \Phi_{\alpha,i}(t)A^{T} + \Sigma_{\mathbf{w}}\nonumber \\
    \label{eq:dynamics-x-hat3}
    &- \Phi_{\alpha,i}(t)C_{\alpha,i}^{T}\Sigma_{\mathbf{v}_{\alpha,i}}^{-1}C_{\alpha,i}\Phi_{\alpha,i}(t)^{T}
\end{align}
where $\Phi_{\alpha,i}(0)$ and $\hat{\mathbf{x}}_{\alpha,i}(0)$ are given.
From~\cite{anderson2007optimal}, the optimal LQG control based on $\hat{\mathbf{x}}_{\alpha,i}(t)$ is
\begin{align}
    \label{eq:controller-attack-free1}
    \mathbf{u}_{\alpha,i}(t) &= \frac{1}{2}K(t)\hat{\mathbf{x}}_{\alpha,i}(t) - \frac{1}{2}R^{-1}B^{T}\mathbf{s}(t) \\
    \label{eq:controller-attack-free2}
    K(t) &= -R^{-1}B^{T}P(t) \\
    \label{eq:controller-attack-free3}
    -\dot{P}(t) &= A^{T}P(t) + P(t)A - \frac{1}{2}P(t)BR^{-1}B^{T}P(t) + 2Q \\
    \label{eq:controller-attack-free4}
    \dot{\mathbf{s}}(t) &= (-A^{T} + \frac{1}{2}P(t)BR^{-1}B^{T})\mathbf{s}(t) + 2Q\mathbf{r}(t)    
\end{align}
where $\mathbf{s}(t)$ and $P(t)$ have boundary conditions $\mathbf{s}(T) = -2Fr(T)$ and $P(T) = 2F$.

Denote $\mathbf{\hat{x}}_{\alpha}^{\ast}(t)$ as the KF estimate of $\mathbf{x}(t)$ based on $\mathbf{y}_{\alpha}^{\ast}(t)$ of sensors in $\{1,\ldots,p\} \setminus \mathcal{A}^{\ast}.$ Dynamics of $\mathbf{\hat{x}}_{\alpha}^{\ast}(t)$ is analogous to Equations~(\ref{eq:dynamics-x-hat1})-(\ref{eq:dynamics-x-hat3}). Similarly, we define $\mathbf{u}_{\alpha}^{\ast}(t)$ as the LQG tracking optimal control input based on $\{\mathbf{y}_{\alpha}^{\ast}(t^{\prime}) \ t^{\prime} \leq t\}$.

Define the set of feasible control inputs at time $t$ with respect to attack pattern $\mathcal{A}_i$ as 
$\mathcal{U}_{\gamma_i}(t) \triangleq \{\mathbf{u}(t) : \left(\mathbf{u}(t) - \mathbf{u}_{\alpha,i}(t)\right)^T\left(\mathbf{u}(t) - \mathbf{u}_{\alpha,i}(t)\right) \leq \gamma_i^2\},$
where $\gamma_i \geq 0$ is a parameter that will be discussed in Section~\ref{subsec:barrier-certificate}. Define $\mathcal{U}_{\gamma}^{\ast}(t) \triangleq \{\mathbf{u}(t) : \|\mathbf{u}(t)-\mathbf{u}_{\alpha}^{\ast}(t)\|_{2} \leq \gamma^{\ast}\}$ and $\mathcal{U}(t) \triangleq \bigcap_{i = 1}^{q} \mathcal{U}_{\gamma_i}(t).$ Using this constraint instead of the constraint in (\ref{eq:problem-form}), the problem becomes
\begin{subequations}\label{eq:problem-form-revised multi}
 \begin{align}
 \min_{\mathbf{u}(t)} \ & \mathbf{E}[\int_{0}^{T}{((\mathbf{x}(t)-\mathbf{r}(t))^{T}Q(\mathbf{x}(t)-\mathbf{r}(t)) }   + \mathbf{u}(t)^{T}R\mathbf{u}(t))  dt \nonumber \\
 &+(\mathbf{x}(T)-\mathbf{r}(T))^{T}F(\mathbf{x}(T)-\mathbf{r}(T)) ] \label{eq:problem-form-revised obj multi}\\
 \mbox{s.t.} \ & \mathbf{u}(t) \in \mathcal{U}(t), \forall t \in [0,T]. \label{eq:problem-form-revised con multi}
 \end{align}
 \end{subequations} 

The solution of (\ref{eq:problem-form-revised multi}) can be computed by solving a stochastic Hamilton-Jacobi-Bellman (HJB) equation~\cite{kirk2004optimal}
\begin{equation}
\label{eq:HJB multi}
    \begin{array}{ll}
        0 &= \min_{\mathbf{u}(t) \in \mathcal{U}(t)}{\left\{(\mathbf{x}(t)-\mathbf{r}(t))^{T}Q(\mathbf{x}(t)-\mathbf{r}(t)) \right.} \\
        &+ V_{\mathbf{x}}(t,\mathbf{x})(A\mathbf{x}(t)+B\mathbf{u}(t)) + \frac{1}{2}\mathbf{tr}(V_{\mathbf{x}\mathbf{x}}(t,\mathbf{x})\Sigma_{\mathbf{w}})\\
        &+ \left.\mathbf{u}(t)^{T}R\mathbf{u}(t) + V_{t}(t,\mathbf{x})  \right\},
    \end{array}
\end{equation}
where the optimal $\mathbf{u}(t)$ with respect to Equation~(\ref{eq:HJB multi}) is equal to the minimizer of Equation~(\ref{eq:HJB multi}) for all $t \in [0, T].$

Solving the constrained partial differential equation (PDE) (\ref{eq:HJB multi}) is challenging, so we relax the constraint of problem \eqref{eq:HJB multi}, and approximate the value function of~(\ref{eq:HJB multi}) by relaxing the constraint~\eqref{eq:problem-form-revised con multi}. We observe that, while we relax~\eqref{eq:problem-form-revised con multi} when approximating the value function, the input will still satisfy $\mathbf{u}(t) \in \mathcal{U}(t).$ The value function is equal to \cite{anderson2007optimal}
\begin{equation}
\label{eq:value-func}
    V(t,\mathbf{x}) = \frac{1}{2}\mathbf{x}(t)^{T}P(t)\mathbf{x}(t) + \beta(t) + \mathbf{s}(t)^{T}\mathbf{x}(t) + s_{0}(t),
\end{equation}
where $\dot{s}_{0}(t) = \frac{1}{4}\mathbf{s}(t)^{T}BR^{-1}B^{T}\mathbf{s}(t) - \mathbf{r}(t)^{T}Q\mathbf{r}(t)$ and $-\dot{\beta}(t) = \frac{1}{2}\mathbf{tr}(P(t)\Sigma_{\mathbf{w}}).$

Substituting Equation~(\ref{eq:value-func}) into Equation~(\ref{eq:HJB multi}), we have 
\begin{multline}
\label{eq:HJB-revised}
 0 = \min_{\mathbf{u}(t) \in \mathcal{U}(t)}\Big\{(\mathbf{x}(t)-\mathbf{r}(t))^{T}Q(\mathbf{x}(t)-\mathbf{r}(t)) + \dot{s}_{0}(t)+\dot{\beta}(t)\\
 + \mathbf{u}(t)^{T}R\mathbf{u}(t) + \mathbf{x}(t)^{T}P(t)(A\mathbf{x}(t) + B\mathbf{u}(t)) + \mathbf{x}(t)^{T}\dot{\mathbf{s}}(t)\\
 + \frac{1}{2}\mathbf{x}(t)^{T}\dot{P}(t)\mathbf{x}(t) + \mathbf{s}(t)^{T}(A\mathbf{x}(t)+B\mathbf{u}(t)) \Big\}
 \end{multline}
We approximate the optimal $\mathbf{u}(t)$ with respect to Equation~(\ref{eq:HJB multi}) by the minimizer of Equation~(\ref{eq:HJB-revised}). Computing the minimizer of Equation~(\ref{eq:HJB-revised}) is equivalent to solving a QCQP
 \begin{align}
  \label{eq:approx_opt}
  \begin{split}
      \min_{\mathbf{u}(t)} \ & \mathbf{u}(t)^{T}R\mathbf{u}(t) + \hat{\mathbf{x}}(t)^{T}P(t)B\mathbf{u}(t) + \mathbf{s}(t)^{T}B\mathbf{u}(t) \\
 \mbox{s.t.} \ & \mathbf{u}(t) \in \mathcal{U}(t), \forall t \in [0,T].
  \end{split}
 \end{align}
at each time $t$. QCQPs in the form of Equation~(\ref{eq:approx_opt}) can be solved efficiently using existing solvers~\cite{domahidi2012efficient}\cite{torrisi2018projected}.

\subsection{Safety and Reachability Verification}
\label{subsec:barrier-certificate}

Parameters $\gamma_i$ in $\mathcal{U}_{\gamma_i} = \{\mathbf{u}(t) : \|\mathbf{u}(t)-\mathbf{u}_{\alpha,i}(t)\|_{2} \leq \gamma_i\}$ determines the size of the set of feasible control inputs at each time $t$. Larger $\gamma_i$ provide more choices of control input, which improve the performance of the system in the attack-free scenario. However, enlarging the feasible control input set also increases the probability that the system may be biased and led to the unsafe states. Thus, there is a tradeoff between the performance and the risk of violating safety when selecting $\gamma_i$.

We develop a binary search algorithm to find the maximal feasible $\gamma_i$ which satisfie the safety and reachability constraints in equation~(\ref{eq:problem-form}). We use the barrier function method to determine whether safety and reachability are guaranteed for each value of $\gamma_i.$ The idea is to construct a barrier function $D_i(\mathbf{x})$ for each $\gamma_i$ such that, for some $L<K,$ $D_i(\mathbf{x_0}) \leq L$, $D_i(\mathbf{x}) > K$ for all $\mathbf{x}(t) \in U$, and $D_i(\mathbf{x})$ is decreasing over any feasible trajectories of $\mathbf{x}(t).$ Thus, if this $D_i(\mathbf{x})$ exists for each $\gamma_i$, $\mathbf{x}(t)$ will not enter the unsafe region.

Let $\hat{\mathbf{u}}_i(t) = \mathbf{u}(t) - \mathbf{u}_{\alpha,i}(t)$ be regarded as the disturbance introduced by $\mathbf{a}(t)$ with respect to each attack pattern $\mathcal{A}_i.$ In order to ensure safety and reachability under any FDI, we assume that $\hat{\mathbf{u}}_i(t)$ could be arbitrary values in $\mathcal{U}_{\gamma_i}(t)$. The dynamics can be rewritten as $\dot{\mathbf{x}}(t) = A\mathbf{x}(t) + B\mathbf{u}_{\alpha,i}(t) + B\hat{\mathbf{u}}_i(t) + \mathbf{w}(t)$. From equation~(\ref{eq:controller-attack-free1}) we know that $\mathbf{u}_{\alpha,i}(t)$ is computed via $\hat{\mathbf{x}}_{\alpha,i}(t).$ Hence, in order to consider the dynamics of both $\mathbf{x}(t)$ and $\hat{\mathbf{x}}_{\alpha,i}(t),$ we develop an extended system
\begin{equation*}
    \begin{array}{ll}
        \dot{\mathbf{x}}(t)&= A\mathbf{x}(t) + \frac{1}{2}BK(t)\hat{\mathbf{x}}_{\alpha,i}(t) + B\hat{\mathbf{u}}_i(t) + \mathbf{w}(t) \\
        &\quad- \frac{1}{2}BR^{-1}B^{T}\mathbf{s}(t) \\
        \dot{\hat{\mathbf{x}}}_{\alpha,i}(t)&= \Theta_{\alpha,i}(t)C_{\alpha,i}\mathbf{x}(t) + A\hat{\mathbf{x}}_{\alpha,i}(t) + B\hat{\mathbf{u}}_i(t) + \frac{1}{2}BK(t)\hat{\mathbf{x}}_{\alpha,i}(t)\\
        &\quad-\Theta_{\alpha,i}(t)C_{\alpha,i}\hat{\mathbf{x}}_{\alpha,i}(t)  + \Theta_{\alpha,i}(t)\mathbf{v}_{\alpha,i}(t)- \frac{1}{2}BR^{-1}B^{T}\mathbf{s}(t)\\
        \Theta_{\alpha,i}(t)&= \Phi_{\alpha,i}(t)C_{\alpha,i}^{T}\Sigma_{\mathbf{v}_{\alpha,i}}^{-1} \\
        \dot{\Phi}_{\alpha,i}(t)&= A\Phi_{\alpha,i}(t) + \Phi_{\alpha,i}(t)A^{T} + \Sigma_{\mathbf{w}} \\
        &\quad- \Phi_{\alpha,i}(t)C_{\alpha,i}^{T}\Sigma_{\mathbf{v}_{\alpha,i}}^{-1}C_{\alpha,i}\Phi_{\alpha,i}(t)^{T}
    \end{array}
\end{equation*}
where $\Phi_{\alpha,i}(0) = 0$ and $\hat{\mathbf{x}}_{\alpha,i}(0)=\mathbf{x}_0$ for $\forall i \in \{1,\ldots,q\}$.

Define $\overline{\mathbf{x}}_i(t) = \big(
\begin{smallmatrix}
\mathbf{x}(t) \\
\hat{\mathbf{x}}_{\alpha,i}(t)
\end{smallmatrix}
\big)$ as an extended state vector, with $$\dot{\overline{\mathbf{x}}}_i(t) = \overline{A}_i\overline{\mathbf{x}}_i(t) + \overline{B}\hat{\mathbf{u}}_i(t) + \overline{F}\mathbf{w}(t) + \overline{G}_i(t)\mathbf{v}_{\alpha,i}(t)+p(t),$$ where
\begin{subequations}
    \begin{align*}
        \overline{A}_i &= \begin{pmatrix}A & \frac{1}{2}BK(t) \\
        \Theta_{\alpha,i}(t)C_{\alpha,i} & A-\Theta_{\alpha,i}(t) C_{\alpha,i} + \frac{1}{2}BK(t)\end{pmatrix},\\
        \overline{B} &= \begin{pmatrix}B\\B\end{pmatrix}, \quad
        \overline{F} = \begin{pmatrix}I \\0 \\\end{pmatrix}, \quad \overline{G} = \begin{pmatrix}0 \\\Theta(t) \\\end{pmatrix},\\
        p(t) &= \begin{pmatrix}    -\frac{1}{2}BR^{-1}B^T\mathbf{s}(t) \\-\frac{1}{2}BR^{-1}B^T\mathbf{s}(t) \end{pmatrix},
    \end{align*}
\end{subequations}
and $\overline{\mathbf{x}}_i(t_0) = \big(
\begin{smallmatrix}
\mathbf{x}_0 \\
\mathbf{x}_0
\end{smallmatrix}
\big) = \overline{\mathbf{x}}_0$. The set of unsafe states of the extended system $\overline{U}_i$ is defined as 
$\overline{U}_i \triangleq \{\big(
\begin{smallmatrix}
\mathbf{x} \\
\hat{\mathbf{x}}_{\alpha,i}
\end{smallmatrix}
\big): \mathbf{x}(t) \in U, t \in [0, T]\}.$
The set of goal states of the extended system $\overline{G}_i$ is defined as 
$\overline{G}_i \triangleq \{\big(
\begin{smallmatrix}
\mathbf{x} \\
\hat{\mathbf{x}}_{\alpha,i}
\end{smallmatrix}
\big): \mathbf{x}(T) \in G\}.$
Let $N_{\mathbf{w}}$ and $N_{\mathbf{v}_{\alpha,i}}$ be matrices that satisfy $N_{\mathbf{w}}N_{\mathbf{w}}^{T} = \Sigma_{\mathbf{w}}$ and $N_{\mathbf{v}_{\alpha,i}}N_{\mathbf{v}_{\alpha,i}}^{T} = \Sigma_{\mathbf{v}_{\alpha,i}}$. Define $\Lambda_i = \big(
\begin{smallmatrix}
\overline{F}N_{\mathbf{w}} \\
\overline{G} N_{\mathbf{v}_{\alpha,i}}
\end{smallmatrix}
\big).$

We next analyze the safety property of this approach using the barrier method. This result follows from Proposition 2 and Theorem 15 of \cite{prajna2007framework} and is provided for completeness. As a preliminary, we define the concept of a martingale as follows.

\begin{Definition}
\label{def:martingale}
A continuous random process $(X_{t})$ is a \emph{martingale} if $\mathbf{E}(X_{s} | X_{t}) = X_{t}$ for all $s \geq t$. A \emph{supermartingale} is a random process such that $\mathbf{E}(X_{s} | X_{t}) \leq X_{t}$ for all $s \geq t$. A \emph{submartingale} is a random process such that $\mathbf{E}(X_{s} | X_{t}) \geq X_{t}$ for all $s \geq t$.
\end{Definition}

The probability that a submartingale crosses a particular bound is bounded as follows.

\begin{Lemma}[Doob's Martingale Inequality~\cite{karatzsas1991brownian}]
\label{lemma:doob_martingale}
Let $X_{t}$ be a nonnegative supermartingale. Then for any $T > 0$ and constant $\theta$, $$Pr\left(\sup_{0 \leq t \leq T}{X_{t}} \geq \theta\right) \leq \frac{\mathbf{E}(X_{T})}{\theta}.$$
\end{Lemma}
 
\begin{proposition}
\label{prop:barrier-safety}
Suppose there exists a function $D_i(\overline{\mathbf{x}}_i)$ such that
\begin{align}
\label{eq:barrier-safety1}
D_i(\overline{\mathbf{x}}_0) &\leq \epsilon_s \\
\label{eq:barrier-safety2}
D_i(\overline{\mathbf{x}}_i) &\geq 1 \ \forall \overline{\mathbf{x}}_i(t) \in \overline{U}_i \\
\label{eq:barrier-safety3}
D_i(\overline{\mathbf{x}}_i) &\geq 0 \ \forall \overline{\mathbf{x}}_i \\
\frac{\partial D_i}{\partial \overline{\mathbf{x}}_i}(\overline{A}_i\overline{\mathbf{x}}_i(t) + \overline{B}\hat{\mathbf{u}}_i(t)) + \frac{\partial D_i}{\partial t} &+ \frac{1}{2}\mathbf{tr}(\Lambda^{T}_i\frac{\partial^{2}D_i}{\partial \overline{\mathbf{x}}_i^{2}}\Lambda_i) \nonumber\\
\label{eq:barrier-safety4}
&\leq 0 \ \forall \overline{\mathbf{x}}_i, ||\hat{\mathbf{u}}_i||_{2} \leq \gamma_i^{s}
\end{align}
Then $Pr\left(\bigcup_{t \in [0,T]}{\{\overline{\mathbf{x}}_i(t) \in \overline{U}_i\}}\right) \leq \epsilon_s$.
\end{proposition}

\begin{proof}
According to the definition, the differential generator of the extended system can be written as
\begin{equation*}
    \label{eq:generator}
    AB(\overline{\mathbf{x}}_i) = \frac{\partial D_i}{\partial \overline{\mathbf{x}}_i}(\overline{A}_i\overline{\mathbf{x}}_i(t) + \overline{B}\hat{\mathbf{u}}_i(t)) + \frac{\partial D_i}{\partial t} + \frac{1}{2}\mathbf{tr}\left(\Lambda_i^{T}\frac{\partial^{2}D_i}{\partial\overline{\mathbf{x}}_i^{2}}\Lambda_i\right)
\end{equation*}
Based on Dynkin's formula and inequality~(\ref{eq:barrier-safety4}), we have
\begin{equation*}
    \begin{array}{ll}
        &\quad\mathbf{E}(D_i(\overline{\mathbf{x}}_i(t)) | \overline{\mathbf{x}}_i(s)) \\
        &= D_i(\overline{\mathbf{x}}_i(s))+ \mathbf{E}\left[\int_{s}^{t}{AB(\overline{\mathbf{x}}_i(\tau)) \ d\tau} | \overline{\mathbf{x}}_i(s)\right] \nonumber\\
        &\leq D(\overline{\mathbf{x}}_i(s)), for \ t \geq s.
    \end{array}
\end{equation*}
Thus, $D_i(\overline{\mathbf{x}}_i(t))$ is a supermartingale. By Doob's martingale inequality and (\ref{eq:barrier-safety3}), we get
\begin{equation}
\label{eq:Doob's}
    Pr\left(\sup_{t \in [0,T]}{D_i(\overline{\mathbf{x}}_i(t))} \geq \lambda \ | \ \overline{\mathbf{x}}_i(t_0) \right) \leq \frac{D_i(\overline{\mathbf{x}}_i(t_0))}{\lambda}
\end{equation}
By inequality~(\ref{eq:barrier-safety2}), we have
\begin{align}
\label{eq:unsafe}
    &Pr\left(\bigcup_{t \in [0,T]}{\{\overline{\mathbf{x}}_i(t) \in \overline{U}_i\}}\right) \leq Pr\left(\bigcup_{t \in [0,T]}{\{D_i(\overline{\mathbf{x}}_i(t)) \geq 1\}}\right) \nonumber\\
    = &Pr\left(\sup_{t \in [0,T]}{D_i(\overline{\mathbf{x}}_i(t))} \geq 1\right).
\end{align}
Combining inequalities~(\ref{eq:barrier-safety1}), (\ref{eq:Doob's}), and (\ref{eq:unsafe}) we get
\begin{equation}
    Pr\left(\bigcup_{t \in [0,T]}{\{\overline{\mathbf{x}}_i(t) \in \overline{U}_i\}}\right) \leq \frac{D_i(\overline{\mathbf{x}}_i(t_0))}{1} \leq \epsilon_s.
\end{equation}
\end{proof}

Proposition~\ref{prop:barrier-safety} shows that if $\mathcal{A}^{\ast} = \mathcal{A}_i$ and there exists a barrier function $D_i(\overline{\mathbf{x}}_i)$ which satisfies inequalities~(\ref{eq:barrier-safety1})-(\ref{eq:barrier-safety4}) with given $\gamma^s_i$, the probability for safety of all trajectories starting from $\mathbf{x}_0$ is guaranteed by given lower bound $\epsilon_s$. The barrier function can be calculated via the sum-of-squares (SOS) optimization \cite{prajna2002introducing}. 

In order to select $\gamma^s_i$, the inequalities~(\ref{eq:barrier-safety1})-(\ref{eq:barrier-safety4}) in Proposition~\ref{prop:barrier-safety} need to be revised so that all the constraints have the form of polynomial SOS.

Define $g_{s_i}(\hat{\mathbf{u}}_i) = (\gamma^s_i)^{2}-\|\hat{\mathbf{u}}_i\|_{2}^{2},$
so that $\|\hat{\mathbf{u}}_i\|_{2} \leq \gamma^s_i$ is equivalent to $g_{s_i}(\hat{\mathbf{u}}_i) \geq 0$.

\begin{proposition}
\label{prop:barrier-safety-sos}
Suppose that there exist polynomials $\lambda_{U_i}(\overline{\mathbf{x}}_i)$, $\lambda_{D_i}(\overline{\mathbf{x}}_i, \hat{\mathbf{u}}_i)$, and $D_i(\overline{\mathbf{x}}_i)$ such that the following hold:
\begin{align}
\label{eq:barrier-safety-sos1}
-D_i(\overline{\mathbf{x}}_{0}) + \epsilon_s &\geq 0 \\
\label{eq:barrier-safety-sos2}
D_i(\overline{\mathbf{x}}_i) - 1 - \lambda_{U_i}^{T}(\overline{\mathbf{x}}_i)g_{U_i}(\overline{\mathbf{x}}_i) &\geq 0 \\
\label{eq:barrier-safety-sos3}
D_i(\overline{\mathbf{x}}_i) &\geq 0 \\
\nonumber
-\frac{\partial D_i}{\partial \overline{\mathbf{x}}_i}(\overline{A}_i\overline{\mathbf{x}}_i(t) + \overline{B}\hat{\mathbf{u}}_i(t)) - \lambda_{D_i}^{T}(\overline{\mathbf{x}}_i,\hat{\mathbf{u}}_i)g_{s_i}(\hat{\mathbf{u}}_i&) \\
\label{eq:barrier-safety-sos4}
 - \frac{\partial D_i}{\partial t} - \frac{1}{2}\mathbf{tr}(\Lambda_i^{T}\frac{\partial^{2}D_i}{\partial \overline{\mathbf{x}}_i^{2}}\Lambda_i) &\geq 0 \\
\label{eq:barrier-safety-sos5}
\lambda_{U_i}(\overline{\mathbf{x}}_i) \geq 0, \lambda_{D_i}(\overline{\mathbf{x}}_i,\hat{\mathbf{u}}_i) &\geq 0
\end{align}
Then $Pr\left(\bigcup_{t \in [0,T]}{\{\overline{\mathbf{x}}_i(t) \in \overline{U}_i\}}\right) \leq \epsilon_s$.
\end{proposition}

\begin{proof}
Inequalities~(\ref{eq:barrier-safety-sos1}) and (\ref{eq:barrier-safety-sos3}) imply inequalities~(\ref{eq:barrier-safety1}) and (\ref{eq:barrier-safety3}). If inequalities~(\ref{eq:barrier-safety-sos2}) and (\ref{eq:barrier-safety-sos5}) hold, we have $D_i(\overline{\mathbf{x}}_i) - 1 \geq \lambda_{U_i}^{T}(\overline{\mathbf{x}}_i)g_{U_i}(\overline{\mathbf{x}}_i)  \geq 0.$ This means inequality~(\ref{eq:barrier-safety2}) holds when $\overline{\mathbf{x}}_i(t) \in \overline{U}_i$. If inequalities~(\ref{eq:barrier-safety-sos4}) and (\ref{eq:barrier-safety-sos5}) hold, we get $\frac{\partial D_i}{\partial \overline{\mathbf{x}}_i}(\overline{A}_i\overline{\mathbf{x}}_i(t) + \overline{B}\hat{\mathbf{u}}_i(t)) + \frac{\partial D_i}{\partial t} + \frac{1}{2}\mathbf{tr}(\Lambda_i^{T}\frac{\partial^{2}D_i}{\partial \overline{\mathbf{x}}_i^{2}}\Lambda_i)   \leq - \lambda_{D_i}^{T}(\overline{\mathbf{x}}_i,\hat{\mathbf{u}}_i)g_{s_i}(\hat{\mathbf{u}}_i) \leq 0 $ when $\|\hat{\mathbf{u}}_i\|_{2} \leq \gamma_i,$ which implies that inequality~(\ref{eq:barrier-safety4}) holds. Hence, Proposition~\ref{prop:barrier-safety} holds, and the probability that the extended system state is in the unsafe region is upper-bounded by $\epsilon_s.$ 
\end{proof}

The barrier function $D_i(\overline{\mathbf{x}}_i)$ in inequalities~(\ref{eq:barrier-safety-sos1})-(\ref{eq:barrier-safety-sos5}) can be calculated via SOS optimization. By checking the existence of $D_i(\overline{\mathbf{x}}_i)$ under a given $\gamma^s_i$, whether the $\gamma^s_i$ satisfies safety and reachability constraints or not can be decided. 
The reachability constraint can be regarded as the safety constraint that only need to be kept at the final time step. Thus, time is also regarded as a state variable of a barrier function $D_i^{'}(\overline{\mathbf{x}}_i, t)$. The unsafe region is defined as $(\mathbb{R}^{2n} \setminus \overline{G}_i) \times \{T\}$. The Proposition~\ref{prop:barrier-reachability} can be derived by a similar way with Proposition~\ref{prop:barrier-safety}. 

\begin{proposition}
\label{prop:barrier-reachability}
Suppose there exists a function $D_i^{'}(\overline{\mathbf{x}}_i, t)$ such that
\begin{align}
\label{eq:barrier-reachability1}
D_i^{'}(\overline{\mathbf{x}}(t_0),0) &\leq \epsilon_r \\
\label{eq:barrier-reachability2}
D_i^{'}(\overline{\mathbf{x}}_i,T) &\geq 1 \ \forall \overline{\mathbf{x}}_i \in (\mathbb{R}^{2n} \setminus \overline{G}_i) \\
\label{eq:barrier-reachability3}
D_i^{'}(\overline{\mathbf{x}}_i,t) &\geq 0 \ \forall \overline{\mathbf{x}}_i, t \\
\frac{\partial D_i^{'}}{\partial \overline{\mathbf{x}}_i}(\overline{A}_i\overline{\mathbf{x}}_i(t) + \overline{B}\hat{\mathbf{u}}_i(t)) &+ \frac{\partial D_i^{'}}{\partial t} + \frac{1}{2}\mathbf{tr}(\Lambda_i^{T}\frac{\partial^{2}D_i^{'}}{\partial \overline{\mathbf{x}}_i^{2}}\Lambda_i)\nonumber\\
\label{eq:barrier-reachability4}
&\leq 0 \ \forall \overline{\mathbf{x}}_i,t, \|\hat{\mathbf{u}}_i\|_{2} \leq \gamma_i^{r}
\end{align}
Then $Pr\left({\overline{\mathbf{x}}_i(T) \in \overline{G}_i}\right) \geq 1-\epsilon_r$.
\end{proposition}

Defining $g_{r_i}(\hat{\mathbf{u}}_i),$ $\lambda_{G_i}(\overline{\mathbf{x}}_i),$ and $\lambda_{D_i}^{\prime}(\overline{\mathbf{x}}_i, \hat{\mathbf{u}}_i, t)$ in a similar way, we revise Proposition~\ref{prop:barrier-reachability} and get
\begin{align}
\label{eq:barrier-reachability-sos1}
-D_i^{\prime}(\overline{\mathbf{x}}_{0}, 0) + \epsilon_r &\geq 0 \\
\label{eq:barrier-reachability-sos2}
D_i^{\prime}(\overline{\mathbf{x}}_i,T) - 1 - \lambda_{G_i}^{T}(\overline{\mathbf{x}}_i)g_{G_i}(\overline{\mathbf{x}}_i) &\geq 0 \\
\label{eq:barrier-reachability-sos3}
D_i^{\prime}(\overline{\mathbf{x}}_i,t) &\geq 0 \\
-\frac{\partial D_i^{'}}{\partial \overline{\mathbf{x}}_i}(\overline{A}_i\overline{\mathbf{x}}_i(t) + \overline{B}\hat{\mathbf{u}}_i(t)) - \lambda_{D_i}^{\prime T}(\overline{\mathbf{x}}_i, \hat{\mathbf{u}}_i, t)g_{r_i}(\hat{\mathbf{u}}_i) &- \frac{\partial D_i^{'}}{\partial t}
\nonumber \\
\label{eq:barrier-reachability-sos4}
- \frac{1}{2}\mathbf{tr}\left(\Lambda_i^{T}\frac{\partial^{2}D_i^{\prime}}{\partial \overline{\mathbf{x}}_i^{2}}\Lambda_i\right) &\geq 0 \\
\label{eq:barrier-reachability-sos5}
\lambda_{G_i}(\overline{\mathbf{x}}_i) \geq 0, \lambda_{D_i}^{\prime}(\overline{\mathbf{x}}_i,\hat{\mathbf{u}}_i, t) &\geq 0 
\end{align}

Based on Proposition~\ref{prop:barrier-safety} and \ref{prop:barrier-reachability}, the safety constraint and reachability constraint at each time $t$ are two balls with identical center $\mathbf{u}_{\alpha,i}(t)$ and different radii $\gamma_i^s$ and $\gamma_i^r$
\begin{equation*}
    \begin{array}{cc}
        \text{Safety}&:||\mathbf{u}(t) - \mathbf{u}_{\alpha,i}(t)||_{2} \leq \gamma_i^s \\
        \text{Reachability}&:||\mathbf{u}(t) - \mathbf{u}_{\alpha,i}(t)||_{2} \leq \gamma_i^r 
    \end{array}
\end{equation*}
By staying within a ball with smaller radius, both safety and reachability will be satisfied. The new QCQP is
 \begin{align}
  \label{eq:QCQP-final}
  \begin{split}
      \min_{\mathbf{u}(t)} \ & \mathbf{u}(t)^{T}R\mathbf{u}(t) + \hat{\mathbf{x}}(t)^{T}P(t)B\mathbf{u}(t) + \mathbf{s}(t)^{T}B\mathbf{u}(t) \\
 \mbox{s.t.} \ & \left(\mathbf{u}(t) - \mathbf{u}_{\alpha,i}(t)\right)^T\left(\mathbf{u}(t) - \mathbf{u}_{\alpha,i}(t)\right) \leq \gamma_i^2, i \in \{1,\ldots,q\}
  \end{split}
 \end{align}
where $\gamma_i = \min \{\gamma_i^s, \gamma_i^r\}$. 

We demonstrate the relationship between the variation of $\gamma_i^s$ and the satisfiability of safety as follows.
\begin{Lemma}
\label{lemma:monotonicity-gamma-safety}
    For all $\gamma_i^s > 0$, if there exists a function $D_i(\overline{\mathbf{x}}_i)$ such that inequalities~(\ref{eq:barrier-safety-sos1})-(\ref{eq:barrier-safety-sos5}) hold, then for all $\hat{\gamma}_i^{s} < \gamma_i^s$, $D_i(\overline{\mathbf{x}}_i)$ satisfies inequalities~(\ref{eq:barrier-safety-sos1})-(\ref{eq:barrier-safety-sos5}) as well.
\end{Lemma}
 
Similarly we present the relationship between the variation of $\gamma_i^r$ and the satisfiability of reachability as follows.
 
\begin{Lemma}
\label{lemma:monotonicity-gamma-reachability}
    For $\forall \gamma_i^r > 0$, if there exists a function $D_i^{\prime}(\overline{\mathbf{x}}_i, t)$ such that  inequalities~(\ref{eq:barrier-reachability-sos1})-(\ref{eq:barrier-reachability-sos5}) hold, then for $\forall \hat{\gamma}_i^{r} < \gamma_i^r$, $D_i^{\prime}(\overline{\mathbf{x}}_i, t)$ satisfies inequalities~(\ref{eq:barrier-reachability-sos1})-(\ref{eq:barrier-reachability-sos5}) as well.
\end{Lemma}
 
The results of Lemma~\ref{lemma:monotonicity-gamma-safety} and~\ref{lemma:monotonicity-gamma-reachability} are straightforward, so we omit the proofs for the compactness of the paper. 

By using binary search and the results of Lemma~\ref{lemma:monotonicity-gamma-safety} and~\ref{lemma:monotonicity-gamma-reachability}, we present  Algorithm~\ref{algo:barrier-certificate}, which is $\rho$-optimal to $\gamma_i^s$ and $\gamma_i^r$ (i.e. $|\gamma_i^s - \gamma_{max,i}^s| \leq \rho$ and $|\gamma_i^r - \gamma_{max,i}^r| \leq \rho$), where $\gamma_{max,i}^s$ and $\gamma_{max,i}^r$ are the maximal $\gamma_i^s$ and $\gamma_i^r$ for which there exist $\hat{D}_i(\overline{\mathbf{x}}_i)$ that satisfies inequalities~(\ref{eq:barrier-safety-sos1})-(\ref{eq:barrier-safety-sos5}) and $\hat{D}_i^{\prime}(\overline{\mathbf{x}}_i, t)$ that satisfies inequalities~(\ref{eq:barrier-reachability-sos1})-(\ref{eq:barrier-reachability-sos5}). Here we assume existence of a function SOS\_Feasible that takes a set of SOS constraints as input and returns 1 if there exist polynomials satisfying the constraints and $0$ otherwise.
 
\begin{algorithm}[h]
	\caption{Algorithm for computing the maximum parameter $\gamma_i$ that ensures safety and reachability under given set of compromised sensors $\mathcal{A}_i$.}
	\label{algo:barrier-certificate}
	\begin{algorithmic}[1]
		\Procedure{Barrier\_Certificate($\epsilon_s$, $\epsilon_r$, $\gamma_{0}$, $\mathcal{A}_i$)}{}
		\State \textbf{Input}: worst case probability of violating safety
		property $\epsilon_s$, worst case probability of violating reachability property $\epsilon_r$, initial upper bound of radii for feasible control input sets $\gamma_{0}$, compromised sensor set $\mathcal{A}_i$
		\State \textbf{Output}: radius of feasible control input set that satisfies safety and reachability properties $\gamma_i$
        \State $\underline{\gamma_i^s} \leftarrow 0$, $\overline{\gamma_i^s} \leftarrow \gamma_{0}$
        \State $\underline{\gamma_i^r} \leftarrow 0$, $\overline{\gamma_i^r} \leftarrow \gamma_{0}$
        
        \While{$|\underline{\gamma_i^s}-\overline{\gamma_i^s}| > \rho$}
        \State $\gamma_i^s \leftarrow (\underline{\gamma_i^s} + \overline{\gamma_i^s})/2$
        \State $flag \leftarrow$ SOS\_Feasible(Eq.~(\ref{eq:barrier-safety-sos1}), Eq.~(\ref{eq:barrier-safety-sos2}), Eq.~(\ref{eq:barrier-safety-sos3}), Eq.~(\ref{eq:barrier-safety-sos4}), Eq.~(\ref{eq:barrier-safety-sos5}), $\epsilon_s$, $\mathcal{A}_i$)
        \If{$flag==0$}
        \State $\overline{\gamma_i^s} \leftarrow \gamma_i^s$
        \Else
        \State $\underline{\gamma_i^s} \leftarrow \gamma_i^s$
        \EndIf
        \EndWhile
        
        \While{$|\underline{\gamma_i^r}-\overline{\gamma_i^r}| > \rho^{\prime}$}
        \State $\gamma_i^r \leftarrow (\underline{\gamma_i^r} + \overline{\gamma_i^r})/2$
        \State $flag \leftarrow$ SOS\_Feasible(Eq.~(\ref{eq:barrier-reachability-sos1}), Eq.~(\ref{eq:barrier-reachability-sos2}), Eq.~(\ref{eq:barrier-reachability-sos3}), Eq.~(\ref{eq:barrier-reachability-sos4}), Eq.~(\ref{eq:barrier-reachability-sos5}), $\epsilon_r$, $\mathcal{A}_i$)
        \If{$flag==0$}
        \State $\overline{\gamma_i^r} \leftarrow \gamma_i^r$
        \Else
        \State $\underline{\gamma_i^r} \leftarrow \gamma_i^r$
        \EndIf
        \EndWhile
        
        \State $\gamma_i \leftarrow \min \{\gamma_i^s, \gamma_i^r\}$
        \State \Return{$\gamma_i$}
            \EndProcedure
	\end{algorithmic}
\end{algorithm}

Since the controller does not know which $\mathcal{A}_i$ is $\mathcal{A}^{\ast},$ we let the control input $\mathbf{u}(t) \in \mathcal{U}(t) = \bigcap_{i = 1}^{q} \mathcal{U}_{\gamma_i}(t)$ to guarantee safety and reachability for all attack patterns $\{\mathcal{A}_i\}.$ However, it is possible that $\mathcal{U}(t) = \emptyset.$ Thus, we need a mechanism to find out feasible solutions when $\mathcal{U}(t) = \emptyset.$

\subsection{Selection of Constraints}
In this subsection, we present a policy to provide feasible $\mathbf{u}(t)$ when $\mathcal{U}(t) = \emptyset.$ Denote $\mathcal{I}(t)$ as the set of the indexes of the constraints $\mathcal{U}_{\gamma_i}(t).$ Define $\gamma_{min} = \min_{i} \gamma_i$, $i \in \{1,\ldots,q\}$. 
In order to identify those $\mathcal{U}_{\gamma_i}(t)$ which cause $\mathcal{U}(t) = \emptyset$, we first give a sufficient condition that $\mathcal{U}(t) \neq \emptyset$. We then express the sufficient condition in terms of the state estimates, and provide a method to select $\mathcal{I}(t)$ such that $\bigcap_{i \in \mathcal{I}(t)}{\mathcal{U}_{\gamma_{i}}(t)} \neq \emptyset$.

\begin{proposition}
\label{prop:no_conflict}
If there exists a ball of radius $\gamma_{min}$ such that $\mathbf{u}_{\alpha,i} \ i \in \{1,\ldots,q\}$ are contained in the ball, then $\mathcal{U}(t) \neq \emptyset.$
\end{proposition}

\begin{proof}
Suppose there exists such a ball with center $\mathbf{u}_0$. We have $\|\mathbf{u}_0 - \mathbf{u}_{\alpha,i}\|_2 \leq \gamma_{min} \leq \gamma_{i} \ \forall i \in \{1,\ldots,q\}$. Hence, $\mathbf{u}_0 \in \mathcal{U}(t)$, and $\mathcal{U}(t) \neq \emptyset.$
\end{proof}

In the following proposition we show the sufficient condition that the ball in Proposition~\ref{prop:no_conflict} exists, and thus $\mathcal{U}(t) \neq \emptyset$.

\begin{proposition}
\label{prop:ball_exist}
For all $i,j \in \{1,\ldots,q\}$, denote $\{\hat{i}, \hat{j}\} = \argmax_{i,j} \{d_{i,j}\}$, where $d_{i,j} = \|\mathbf{u}_{\alpha,i}-\mathbf{u}_{\alpha,j}\|_2$. If $d_{\hat{i}, \hat{j}} > 2\gamma_{min}$, the ball which satisfies Proposition~\ref{prop:no_conflict} does not exist. If $d_{\hat{i}, \hat{j}} \leq 2\gamma_{min}$, and $\|\mathbf{u}_{\alpha,k} - (\mathbf{u}_{\alpha,\hat{i}} + \mathbf{u}_{\alpha,\hat{j}})/2\|_2 \leq \gamma_{min} \ \forall k \in \{1,\ldots,q\} \setminus \{\hat{i}, \hat{j}\}$ holds, then there exists a ball that satisfies Proposition~\ref{prop:no_conflict}.
\end{proposition}
 
\begin{proof}
If $d_{\hat{i}, \hat{j}} > 2\gamma_{min}$, then the distance between $\mathbf{u}_{\alpha,\hat{i}}$ and $\mathbf{u}_{\alpha,\hat{j}}$ is greater than the diameter of the ball in Proposition~\ref{prop:no_conflict}. Thus, there does not exist such ball that satisfies Proposition~\ref{prop:no_conflict}. If $d_{\hat{i}, \hat{j}} \leq 2\gamma_{min}$, then $\|\mathbf{u}_{\alpha,\hat{i}} - (\mathbf{u}_{\alpha,\hat{i}} + \mathbf{u}_{\alpha,\hat{j}})/2\|_2 = \|\mathbf{u}_{\alpha,\hat{j}} - (\mathbf{u}_{\alpha,\hat{i}} + \mathbf{u}_{\alpha,\hat{j}})/2\|_2 = \frac{1}{2}\|\mathbf{u}_{\alpha,\hat{i}} - \mathbf{u}_{\alpha,\hat{j}}\|_2 = \frac{1}{2}d_{\hat{i}, \hat{j}} \leq \gamma_{min}$. Since we also have $\|\mathbf{u}_{\alpha,k} - (\mathbf{u}_{\alpha,\hat{i}} + \mathbf{u}_{\alpha,\hat{j}})/2\|_2 \leq \gamma_{min} \ \forall k \in \{1,\ldots,q\} \setminus \{\hat{i}, \hat{j}\}$, we have that $\mathbf{u}_{\alpha,i}$ are in the ball with center $(\mathbf{u}_{\alpha,\hat{i}} + \mathbf{u}_{\alpha,\hat{j}})/2$ and radius $\gamma_{min}$ for all $i \in \{1,\ldots,q\}$.
\end{proof}

Proposition~\ref{prop:no_conflict} and \ref{prop:ball_exist} imply that if $\|\mathbf{u}_{\alpha,i} - (\mathbf{u}_{\alpha,\hat{i}} + \mathbf{u}_{\alpha,\hat{j}})/2\|_2 \leq \gamma_{min} \ \forall i \in \{1,\ldots,q\}$, then $\mathcal{U}(t) \neq \emptyset.$ 
By definition of $\mathbf{u}_{\alpha,i}(t),$ we rewrite $\|\mathbf{u}_{\alpha,i} - (\mathbf{u}_{\alpha,\hat{i}} + \mathbf{u}_{\alpha,\hat{j}})/2\|_2 \leq \gamma_{min} \ \forall i \in \{1,\ldots,q\}$ as
\begin{multline}
    \label{eq:triangle-inequality}
    \|\frac{1}{4}K(t)(\mathbf{\hat{x}}_{\alpha,i}(t) - \mathbf{\hat{x}}_{\alpha,\hat{i}}(t)) + \frac{1}{4}K(t)(\mathbf{\hat{x}}_{\alpha,i}(t) - \mathbf{\hat{x}}_{\alpha,\hat{j}}(t))\|_2 \\ \leq \gamma_{min}.
\end{multline}
where $K(t)$ is the KF gain. 

Hence, if $\mathcal{U}(t) = \emptyset$, then 
\begin{multline}
\label{eq:triangle-inequality-1}
    \|\frac{1}{4}K(t)(\mathbf{\hat{x}}_{\alpha,i}(t) - \mathbf{\hat{x}}_{\alpha,\hat{i}}(t)) + \frac{1}{4}K(t)(\mathbf{\hat{x}}_{\alpha,i}(t) - \mathbf{\hat{x}}_{\alpha,\hat{j}}(t))\|_2 \\
        > \gamma_{min}.
\end{multline}
In the next lemma, we split~(\ref{eq:triangle-inequality-1}) into two inequalities, with each containing only two estimates.

\begin{Lemma}
\label{lemma:triange-inequality}
If $\ \mathcal{U}(t) = \emptyset$, then either $\ \|\frac{1}{4}K(t)(\mathbf{\hat{x}}_{\alpha,i}(t) - \mathbf{\hat{x}}_{\alpha,\hat{i}}(t))\|_2 > \frac{1}{2}\gamma_{min}$ or $\|\frac{1}{4}K(t)(\mathbf{\hat{x}}_{\alpha,i}(t) - \mathbf{\hat{x}}_{\alpha,\hat{j}}(t))\|_2 > \frac{1}{2}\gamma_{min}$, or both of them hold. 
\end{Lemma}

\begin{proof}
Applying triangle inequality to the left hand side of inequality~(\ref{eq:triangle-inequality-1}), we then have
\begin{align}
        &\|K(t)(\mathbf{\hat{x}}_{\alpha,i}(t) - \mathbf{\hat{x}}_{\alpha,\hat{i}}(t))/4\|_2+ \|K(t)(\mathbf{\hat{x}}_{\alpha,i}(t) - \mathbf{\hat{x}}_{\alpha,\hat{j}}(t))/4\|_2 \nonumber\\
        &\geq \left\|\frac{K(t)(\mathbf{\hat{x}}_{\alpha,i}(t) - \mathbf{\hat{x}}_{\alpha,\hat{i}}(t))}{4} + \frac{K(t)(\mathbf{\hat{x}}_{\alpha,i}(t) - \mathbf{\hat{x}}_{\alpha,\hat{j}}(t))}{4}\right\|_2 \nonumber\\
        &> \gamma_{min}.
        \label{eq:triangle-inequality-2}
\end{align}
Inequality~(\ref{eq:triangle-inequality-2}) is satisfied if at least one of $\|\frac{1}{4}K(t)(\mathbf{\hat{x}}_{\alpha,i}(t) - \mathbf{\hat{x}}_{\alpha,\hat{i}}(t))\|_2 > \frac{1}{2}\gamma_{min}$ and $\|\frac{1}{4}K(t)(\mathbf{\hat{x}}_{\alpha,i}(t) - \mathbf{\hat{x}}_{\alpha,\hat{j}}(t))\|_2 > \frac{1}{2}\gamma_{min}$ is satisfied.
\end{proof}

Motivated by Lemma~\ref{lemma:triange-inequality}, our approach to selecting $\mathcal{I}(t)$ such that $\bigcap_{i \in {\mathcal{I}}(t)} \mathcal{U}_{\gamma_i} \neq \emptyset$ is to compare between two state estimates. Next, we show this comparison.

Denote $C_{\alpha, i, j}$ as $C$ with rows indexed by $\{1,\ldots,p\} \setminus (\mathcal{A}_{i} \bigcup \mathcal{A}_{j})$. We assume that all systems $(A, C_{\alpha, i, j}) \ \forall i,j \in \{1,\ldots,q\} \ i \neq j$ are observable. Introduce the KF state estimate $\hat{\mathbf{x}}_{\alpha,i,j}(t)$, which is obtained via $\mathbf{y}_{\alpha,i,j}(t)$, the output with the measurements indexed by $\{1,\ldots,p\} \setminus (\mathcal{A}_i \bigcup \mathcal{A}_j)$. 

\begin{Lemma}
\label{lemma:triange-inequality2}
If $\|\frac{1}{4}K(t)(\mathbf{\hat{x}}_{\alpha,i}(t) - \mathbf{\hat{x}}_{\alpha,j}(t)\|_2 > \frac{1}{2}\gamma_{min},$ then at least one of $\|\frac{1}{4}K(t)(\mathbf{\hat{x}}_{\alpha,i}(t) - \mathbf{\hat{x}}_{\alpha,i,j}(t)\|_2 > \frac{1}{4}\gamma_{min}$ and $\|\frac{1}{4}K(t)(\mathbf{\hat{x}}_{\alpha,j}(t) - \mathbf{\hat{x}}_{\alpha,i,j}(t)\|_2 > \frac{1}{4}\gamma_{min}$ holds.
\end{Lemma}

\begin{proof}
Applying triangle inequality, we have
    \begin{align}
        \gamma_{min}/2 < &\|K(t)(\mathbf{\hat{x}}_{\alpha,i}(t) - \mathbf{\hat{x}}_{\alpha,j}(t))/4\|_2 \nonumber \\
        \leq &\|K(t)(\mathbf{\hat{x}}_{\alpha,i}(t) - \mathbf{\hat{x}}_{\alpha,i,j}(t))/4\|_2\nonumber \\
        &+ \|K(t)(\mathbf{\hat{x}}_{\alpha,j}(t) - \mathbf{\hat{x}}_{\alpha,i,j}(t))/4\|_2 
         \label{eq:triangle-inequality-3}
    \end{align}
In order for inequality~(\ref{eq:triangle-inequality-3}) to hold, at least one of the following inequalities holds:
\begin{align}
    \|\frac{1}{4}K(t)(\mathbf{\hat{x}}_{\alpha,i}(t) - \mathbf{\hat{x}}_{\alpha,i,j}(t)\|_2 &> \frac{1}{4}\gamma_{min} \label{eq:metric1}\\
    \|\frac{1}{4}K(t)(\mathbf{\hat{x}}_{\alpha,j}(t) - \mathbf{\hat{x}}_{\alpha,i,j}(t)\|_2 &> \frac{1}{4}\gamma_{min}
    \label{eq:metric2}
\end{align}
\end{proof}

We use inequalities~\eqref{eq:metric1}-\eqref{eq:metric2} later to identify the $\mathcal{U}_{\gamma_i}$ that leads to infeasibility of QCQP~(\ref{eq:QCQP-multiple-ad2}).
Intuitively, for a certain pair of $\{i, j\} \in \{1,\ldots,q\},$ if the measurements are only affected by the noises, $\|\mathbf{\hat{x}}_{\alpha,i}(t) - \mathbf{\hat{x}}_{\alpha,j}(t)\|_2$ should be smaller than some thresholds. If $\mathcal{A}_i = \mathcal{A}^{\ast}$ or $\mathcal{A}_j = \mathcal{A}^{\ast},$ $\hat{\mathbf{x}}_{\alpha,i,j}(t)$ should not be biased by $\mathbf{a}(t).$ Thus, when $\|\frac{1}{4}K(t)(\mathbf{\hat{x}}_{\alpha,i}(t) - \mathbf{\hat{x}}_{\alpha,j}(t)\|_2 > \frac{1}{2}\gamma_{min}$, $\hat{\mathbf{x}}_{\alpha,i,j}(t)$ can be utilized as a benchmark for checking whether $\mathbf{\hat{x}}_{\alpha,i}(t)$ and $\mathbf{\hat{x}}_{\alpha,j}(t)$ are affected by the attack and diverge from the unaffected values.

Since both the noise and attack may result in the divergence between state estimates, it is necessary to determine the worst case probability that the noise results in $\|\frac{1}{4}K(t)(\mathbf{\hat{x}}_{\alpha,j}(t) - \mathbf{\hat{x}}_{\alpha,i,j}(t)\|_2 > \frac{1}{4}\gamma_{min}$, which could result in measurements being excluded erroneously. We derive the following theorem to show the probability that $\|\frac{1}{4}K(t)(\mathbf{\hat{x}}_{\alpha,j}(t) - \mathbf{\hat{x}}_{\alpha,i,j}(t)\|_2 > \frac{1}{4}\gamma_{min}$ $\forall i,j \in \{1,\ldots,q\}$ happens during $t \in [0, T]$ is upper-bounded when no adversary is present. We will utilize this theorem later to eliminate $\mathcal{U}_{\gamma_i}(t)$ which may render $\mathcal{U}(t) = \emptyset$.

\begin{Theorem}
\label{theorem:stability_guarantee_EFK_attack_free}
Suppose $\mathcal{A}^{\ast} = \mathcal{A}_i$. There exists $\eta^{i,j}$ such that for each $j \in \{1,\ldots,q\} \setminus \{i\}$ $$Pr(\sup_{t \in [0,T]} \| K(t)\left(\mathbf{\hat{x}}_{\alpha,i}(t) - \mathbf{\hat{x}}_{\alpha,i,j}(t)\right) \|_2 > \gamma_{min} ) \leq \eta^{i,j},$$ where $\eta^{i,j} = {4(\lambda_i^{\ast}\Gamma_i + \lambda_{i,j}^{\ast}\Gamma_{i,j})\overline{K}^{2}}/{\gamma_{min}^{2}},$ $\mathbf{\hat{x}}_{\alpha,i}(t)$ and $\mathbf{\hat{x}}_{\alpha,i,j}(t)$ are estimates calculated using KF and measurements of sensors indexed by $\{1,\ldots,p\} \setminus \mathcal{A}_i$ and $\{1,\ldots,p\} \setminus (\mathcal{A}_i \bigcup \mathcal{A}_j)$, respectively, $\overline{K} = \sup_{t \in [0,T]} ||K(t)||_2,$ $\lambda_i^{\ast} = \sup_{t \in [0,T]}{\{\lambda_{max}(\Sigma_i(t))\}},$ $\lambda_{i,j}^{\ast} = \sup_{t \in [0,T]}{\{\lambda_{max}(\Sigma_{i,j}(t))\}},$ $\Sigma_i(t)$ and $\Sigma_{i,j}(t)$ are the covariance matrices of $(\mathbf{x}(t)-\hat{\mathbf{x}}_{\alpha,i}(t))$ and $(\mathbf{x}(t)-\hat{\mathbf{x}}_{\alpha,i,j}(t))$, respectively, $\lambda_{max}(\cdot)$ denotes the maximum eigenvalue of a matrix, $\Gamma_i = \mathbf{E}\left(\mathbf{e}_i(0)^T{\Sigma_i(0)}^{-1}\mathbf{e}_i(0)\right)$, $\Gamma_{i,j} = \mathbf{E}\left(\mathbf{e}_{ij}(0)^T{\Sigma_{ij}(0)}^{-1}\mathbf{e}_{ij}(0)\right),$ $\mathbf{e}_i(0) = \hat{\mathbf{x}}_{\alpha,i}(0) - \mathbf{x}_0,$ and $\mathbf{e}_{ij}(0) = \hat{\mathbf{x}}_{\alpha,i,j}(0) - \mathbf{x}_0.$
\end{Theorem}
\begin{proof}
Please see the appendix for the detailed proof.
\end{proof}

Define $\{||K(t)\left(\mathbf{\hat{x}}_{\alpha,i}(t) - \mathbf{\hat{x}}_{\alpha,i,j}(t)\right) ||_2 > \gamma_{min}\}$ as $\Omega^{ij}_{i},$ $\{||K(t)\left(\mathbf{\hat{x}}_{\alpha,i}(t) - \mathbf{\hat{x}}_{\alpha,i,j}(t)\right) ||_2 \leq \gamma_{min}\}$ as $\overline{\Omega}^{ij}_{i}.$ Theorem~\ref{theorem:stability_guarantee_EFK_attack_free} implies that the probability that $\Omega^{ij}_{i}$ occurs during $t \in [0, T]$ is bounded above by $\eta^{i,j}$ $\forall i,j \in \{1,\ldots,p\}$ when $\mathcal{A}^{\ast} = \mathcal{A}_i$. In other words, we can eliminate $\hat{\mathbf{x}}_i(t)$ if $\Omega_i^{ij}$ occurs, and the probability that we improperly eliminate an uncompromised estimate $\hat{\mathbf{x}}_i(t)$ ($\mathcal{A}^{\ast} = \mathcal{A}_i$ but we eliminate $\mathcal{U}_{\gamma_i}$) is bounded above by $\eta^{i,j}.$ Applying the results of Propositions~\ref{prop:no_conflict} and \ref{prop:ball_exist} and Theorem~\ref{theorem:stability_guarantee_EFK_attack_free}, we propose the function $\mathcal{I}$\_Selection in Algorithm~\ref{algo:control-policy} to select constraints that can provide feasible control inputs $\mathbf{u}(t)$ to guarantee safety and reachability requirements.

\begin{algorithm}[h]
	\caption{Algorithm for selecting constraints $\mathcal{U}_{\gamma_i}(t)$ that can guarantee safety and reachability properties with desired probability.}
	\label{algo:control-policy}
	\begin{algorithmic}[1]
 		\Procedure{$\mathcal{I}$\_Selection($q$, $K(t)$, $\hat{\mathbf{x}}_{\alpha,i}(t)$, $\hat{\mathbf{x}}_{\alpha,i,j}(t)$, $\gamma_{min}$, $\mathcal{U}_{\gamma_i}(t)$, $i,j \in \{1,\ldots,q\}, i \neq j$)}{}
 		\State \textbf{Input: }number of attack patterns $q$, LQG controller gain $K(t)$, state estimates excluding each attack pattern $\hat{\mathbf{x}}_{\alpha,i}(t)$, state estimates excluding each pair of attack patterns $\hat{\mathbf{x}}_{\alpha,i,j}(t)$, minimum of radii for all constraints $\gamma_{min}$, feasible control input set corresponding to each attack pattern $\mathcal{U}_{\gamma_i}$
 		\State \textbf{Output: }set of indexes of selected constraints $\mathcal{I}(t)$
 		\State $\mathcal{I}(t) \leftarrow \{1,\ldots,q\}$
 		\State $d_{i,j} \leftarrow \|\hat{\mathbf{x}}_{\alpha,i}(t) - \hat{\mathbf{x}}_{\alpha,j}(t)\|_2, \ i,j\in\mathcal{I}(t), \ i\neq j$
 		\State $\{\hat{i}, \hat{j}\} \leftarrow \argmax_{i,j \in \mathcal{I}(t)} \{d_{i,j}\}$
        \While{$d_{\hat{i},\hat{j}} > 2\gamma_{min}$}
            \If{$\Omega^{\hat{i}\hat{j}}_{\hat{i}}$}
                \State $\mathcal{I}(t) \leftarrow \{1,\ldots,q\} \setminus \{\hat{i}\}$
            \EndIf
            
            \If{$\Omega^{\hat{i}\hat{j}}_{\hat{j}}$}
                \State $\mathcal{I}(t) \leftarrow \{1,\ldots,q\} \setminus \{\hat{j}\}$
            \EndIf
            
            \State $\{\hat{i}, \hat{j}\} \leftarrow \argmax_{i,j \in \mathcal{I}(t)} \{d_{i,j}\}$
        \EndWhile
            
 		\ForEach{$i \in \mathcal{I}(t)$}
 		    \If{$||\frac{1}{4}K(t)(\mathbf{\hat{x}}_{\alpha,i}(t) - \mathbf{\hat{x}}_{\alpha,\hat{i}}(t))||_2 > \frac{1}{2}\gamma_{min}$}
 
 		        \If{$\Omega_i^{i\hat{i}}$}
 		            \State $\mathcal{I}(t) \leftarrow \mathcal{I}(t) \setminus \{i\}$
 		        \EndIf
 		            
 		        \If{$\Omega_{\hat{i}}^{i\hat{i}}$}
 		            \State $\mathcal{I}(t) \leftarrow \mathcal{I}(t) \setminus \{\hat{i}\}$
 		        \EndIf
 		    \EndIf
 		    
 		    \If{$||\frac{1}{4}K(t)(\mathbf{\hat{x}}_{\alpha,i}(t) - \mathbf{\hat{x}}_{\alpha,\hat{j}}(t))||_2 > \frac{1}{2}\gamma_{min}$}
 		  
 		        \If{$\Omega_i^{i\hat{j}}$}
 		            \State $\mathcal{I}(t) \leftarrow \mathcal{I}(t) \setminus \{i\}$
 		        \EndIf
 		            
 		        \If{$\Omega_{\hat{j}}^{i\hat{j}}$}
 		            \State $\mathcal{I}(t) \leftarrow \mathcal{I}(t) \setminus \{\hat{j}\}$
 		        \EndIf
 		    \EndIf
 		    
 		    \State $\{\hat{i}, \hat{j}\} \leftarrow \argmax_{i,j \in \mathcal{I}(t)} \{d_{i,j}\}$
 		\EndFor

        \State \Return{$\mathcal{I}(t)$}
        \EndProcedure
	 \end{algorithmic}
\end{algorithm}

Algorithm~\ref{algo:control-policy} works as follows. It requires the number of attack patterns $q$, the LQG controller gain $K(t)$, the state estimates excluding each attack pattern $\hat{\mathbf{x}}_{\alpha,i}(t)$, the state estimates excluding each pair of attack patterns $\hat{\mathbf{x}}_{\alpha,i,j}(t)$, the minimum of radii for all constraints $\gamma_{min}$, the feasible control input set corresponding to each attack pattern $\mathcal{U}_{\gamma_i}$ as the inputs, and returns the set of indexes of selected constraints $\mathcal{I}(t)$ as the output. The algorithm selects constraints $\mathcal{U}_{\gamma_i}(t)$ that can provide feasible control inputs $\mathbf{u}(t)$ to guarantee safety and reachability properties with desired probability. The existence of the feasible control inputs $\mathbf{u}(t)$ is guaranteed by satisfying the sufficient conditions in Proposition~\ref{prop:ball_exist}. Specifically, the condition $d_{\hat{i}, \hat{j}} \leq 2\gamma_{min}$ is guaranteed by line 5 - line 10. The condition $\|\mathbf{u}_{\alpha,k} - (\mathbf{u}_{\alpha,\hat{i}} + \mathbf{u}_{\alpha,\hat{j}})/2\|_2 \leq \gamma_{min} \ \forall k \in \{1,\ldots,q\} \setminus \{\hat{i}, \hat{j}\}$ is verified via line 11 - 22. The judgment statements in line 12 and line 17 select the pairs of state estimates that may be affected by the adversary based on Lemma~\ref{lemma:triange-inequality}. The constraints that are likely to be affected are eliminated in line 13 - 16 and line 18 - 21 based on Lemma~\ref{lemma:triange-inequality2}.

\subsection{Control Strategy Design}
Our proposed control design is summarized in Algorithm~\ref{algo:overall-control-policy}. In line 2 and 3, we initialize $\mathcal{I}(t)$ and $\mathcal{U}(t)$ as the indexes and intersection of all constraints. In line 4, we first check whether $\mathcal{U}(t) = \emptyset$. In line 5, if $\mathcal{U}(t) = \emptyset$, we utilize $\mathcal{I}$\_Selection in Algorithm~\ref{algo:control-policy} to identify and eliminate those $\mathcal{U}_{\gamma_i}(t)$ which result in $\mathcal{U}(t) = \emptyset$ and output $\mathcal{I}(t)$. 
In line 6, the algorithm invokes the existing solver, denoted as $\text{QCQP}(\mathcal{I}(t)),$ at each time $t$ to solve the QCQP with the form 
\begin{equation}
\label{eq:QCQP-multiple-ad2}
\begin{split}
\min_{\mathbf{u}(t)} & \ \mathbf{u}(t)^{T}R\mathbf{u}(t) + \hat{\mathbf{x}}(t)^{T}P(t)B\mathbf{u}(t) + \mathbf{s}(t)^{T}B\mathbf{u}(t) \\
\mbox{s.t.} & \ \mathbf{u}(t) \in \mathcal{U}_{\gamma_i}(t), \ i \in \mathcal{I}(t)
\end{split}
\end{equation}

\begin{algorithm}[h]
	\caption{Proposed control policy that guarantees safety and reachability constraints under multiple-adversary scenario.}
	\label{algo:overall-control-policy}
	\begin{algorithmic}[1]
	\Procedure{Control\_Policy($q$, $\mathcal{U}_{\gamma_i}(t)$, $i\in\{1,\ldots,q\}$)}{}
	    \State \textbf{Input: }number of possible compromised sensor sets $q$, feasible control input set corresponding to each possible compromised sensor set $\mathcal{U}_{\gamma_i}$
	    \State \textbf{Output: }control input at each time step $\mathbf{u}(t)$
	    \State $\mathcal{I}(t) \leftarrow \{1,\ldots,q\}$
	    \State $\mathcal{U}(t) \leftarrow \bigcap_{i \in \mathcal{I}(t)} \mathcal{U}_{\gamma_i}(t)$
	    \If{$\mathcal{U}(t) == \emptyset$}
	        \State $\mathcal{I}(t) \leftarrow \mathcal{I}$\_Selection($q$, $K(t)$, $\hat{\mathbf{x}}_{\alpha,i}(t)$, $\hat{\mathbf{x}}_{\alpha,i,j}(t)$, $\gamma_{min}$, $\mathcal{U}_{\gamma_i}(t)$, $i,j \in \{1,\ldots,q\}, i \neq j$)
	    \EndIf
	    
	    \State $\mathbf{u}(t) \leftarrow \text{QCQP}(\mathcal{I}(t))$ in Equation~(\ref{eq:QCQP-multiple-ad2})
	    \State \Return{$\mathbf{u}(t)$}
	\EndProcedure
	\end{algorithmic}
\end{algorithm}

When there is no attack, the controller attempts to minimize the objective function. Due to the existence of noise, $d_{i,j}$ may deviate from $0$ for $\forall i,j \in \{1,\ldots,q\}.$ This may lead to smaller feasible region $\mathcal{U}(t)$, and suboptimal performance with respect to expected cost. If all $\mathbf{u}_{\alpha,i}(t)$ can be proved to be close to the optimal control $\mathbf{u}^{\ast}(t)$, the feasibility and performance of the proposed approach can be guaranteed.


\begin{Lemma}
\label{lemma:feasibility-multi}
Let $\mathbf{u}^{\ast}(t) = \frac{1}{2}K(t)\hat{\mathbf{x}}(t) - \frac{1}{2}R^{-1}B^{T}\mathbf{s}(t).$ Define $\lambda^{\ast} = \sup_{t}{\{\lambda_{max}(\Sigma(t))\}},$ where $\Sigma(t)$ is the covariance matrix of $(\mathbf{x}(t)-\hat{\mathbf{x}}(t))$. Let $\overline{\eta} = \max\{\eta^i: \forall i \in \mathcal{I}(t)\},$ where $\eta^i = {(\lambda^{\ast}\Gamma + \lambda_{i}^{\ast}\Gamma_i)\overline{K}^{2}}/{\gamma_{min}^{2}}$ and $\Gamma = \mathbf{E}\left(\mathbf{e}(0)^T{\Sigma(0)}^{-1}\mathbf{e}(0)\right).$ When $\beta=1$, we have

\begin{align*}
    Pr(\sup_{t \in [0,T]}{||\mathbf{u}^{\ast}(t)-\mathbf{u}_{\alpha,i}(t)||_{2}} &\leq \gamma_{min}) \geq 1 - \overline{\eta}, \forall i \in \mathcal{I}(t), \\
    Pr\left(\mathbf{u}^{\ast}(t) \in \mathcal{U}(t) \ \forall t \in [0,T] \right) 
    &\geq 1 - \Sigma_{i \in \mathcal{I}(t) }  \ \eta^i
\end{align*}
\end{Lemma}

\begin{proof}
Based on the definitions of $\mathbf{u}^{\ast}(t)$ and $\mathbf{u}_{\alpha,i}(t)$, we have
\begin{equation*}
    \begin{array}{ll}
        &Pr(\sup_{t \in [0,T]}{||\mathbf{u}^{\ast}(t)-\mathbf{u}_{\alpha,i}(t)||_{2}} \geq \gamma_{min}) \\
        = &Pr(\sup_{t \in [0,T]}{||K(t)(\hat{\mathbf{x}}(t)-\hat{\mathbf{x}}_{\alpha,i}(t))||_{2}} \geq 2\gamma_{min})
    \end{array}
\end{equation*}
According to Theorem~\ref{theorem:stability_guarantee_EFK_attack_free}, we have $$Pr(\sup_{t \in [0,T]}{||K(t)(\hat{\mathbf{x}}(t)-\hat{\mathbf{x}}_{\alpha,i}(t))||_{2}} \geq 2\gamma_{min}) \leq \eta^i.$$ Thus,
\begin{equation}
\label{eq:distance-guarantee}
    Pr(\sup_{t \in [0,T]}{||\mathbf{u}^{\ast}(t)-\mathbf{u}_{\alpha,i}(t)||_{2}} \leq \gamma_{min}) \geq 1 - \eta^i
\end{equation}
We have 
\begin{equation}
\label{eq:feasibility-guarantee}
    Pr(\sup_{t \in [0,T]}{||\mathbf{u}^{\ast}(t)-\mathbf{u}_{\alpha,i}(t)||_{2}} \leq \gamma_{min}) \geq 1 - \overline{\eta}, \forall i \in \mathcal{I}(t).
\end{equation}
Based on Equation~\eqref{eq:distance-guarantee} we can obtain
\begin{align}
    &\quad Pr\left(\mathbf{u}^{\ast}(t) \in \mathcal{U}(t) \ \forall t \in [0,T] \right) \nonumber\\
    &= 1 - Pr(\cup_{i \in \mathcal{I}(t)} \mathbf{u}^{\ast}(t) \notin \mathcal{U}_{\gamma_i}(t) \ \forall t \in [0,T]) ) \nonumber\\
    &\geq 1 - \Sigma_{i \in \mathcal{I}(t) }  \ \eta^i \label{eq:optimality-guarantee}
\end{align}
\end{proof}

Based on~\eqref{eq:feasibility-guarantee}, the probability that in the non-adversary case 
$\bigcap_{i \in \mathcal{I}(t)}{\mathcal{U}_{\gamma_{i}}(t)} \neq \emptyset$ is lower bounded by $1 - \overline{\eta}$. Our proposed approach guarantees feasibility under benign environment. Equation~\eqref{eq:optimality-guarantee} implies that the probability that our proposed approach provides the same utility as the best possible control when no adversary is present is lower bounded.

\subsection{Safety and Reachability Guarantees}
In this subsection, we present the safety and reachability guarantees provided by the control policy obtained by Algorithm~\ref{algo:overall-control-policy}. Define $\Omega_{sr} \triangleq (\bigcap_{t \in [0,T]} \{\mathbf{x}(t) \notin U\}) \bigcap \{\mathbf{x}(T) \in G \}$, $\Omega_{\mathcal{U}} \triangleq \bigcap_{t \in [0,T]} \{\mathbf{u}(t) \in \mathcal{U}_{\gamma}^
{\ast}(t)\}$, and $\overline{\Omega}_{\mathcal{U}} \triangleq \bigcup_{t \in [0,T]} \{\mathbf{u}(t) \notin \mathcal{U}_{\gamma}^
{\ast}(t)\}$. The safety and reachability analysis of our proposed control policy is based on bounding the probability $P_0 = Pr\left(\Omega_{sr}\right)$. We define $P_1 \triangleq Pr(\Omega_{sr} | \Omega_{\mathcal{U}})$ as the probability that safety and reachability constraints are satisfied given that the control inputs are from $\mathcal{U}_{\gamma}^{\ast}(t)$. This probability has been discussed in Proposition~\ref{prop:barrier-safety} and \ref{prop:barrier-reachability}. We denote $P_2 \triangleq Pr(\Omega_{\mathcal{U}})$ as the probability that at any time $t$ the control input $\mathbf{u}(t)$ satisfies the correct constraint $\mathcal{U}_{\gamma}^{\ast}(t).$ We have
\begin{equation*}
\label{eqn:safety_problem}
\begin{array}{ll}
P_0 &= Pr\left(\Omega_{sr}\right) \\
&= Pr(\Omega_{sr} | \Omega_{\mathcal{U}})\cdot Pr(\Omega_{\mathcal{U}}) + Pr(\Omega_{sr} | \overline{\Omega}_{\mathcal{U}})\cdot Pr(\overline{\Omega}_{\mathcal{U}}) \\
&\geq P_1 \cdot P_2 \\
\end{array}
\end{equation*}
Here $P_0$ denotes the probability that safety and reachability are guaranteed and the control input $\mathbf{u}(t)$ is in $\mathcal{U}_{\gamma}^{\ast}(t), \forall t,$ and can be expressed using $P_1$ and $P_2$. If there exist lower bounds for both $P_1$ and $P_2$, the lower bound for $P_0$ exists. 

The safety and reachability guarantees of the proposed control policy is presented by the following theorem.

\begin{Theorem}
\label{theorem:control-policy-safety-reachability}
The control strategy returned by Algorithm~\ref{algo:overall-control-policy} satisfies the safety and reachability constraints in Equation~(\ref{eq:problem-form}). 
\end{Theorem}

\begin{proof}
The probability that safety and reachability are guaranteed when the control input $\mathbf{u}(t)$ is in $\mathcal{U}_{\gamma}^{\ast}(t)$ for all $t,$ with probability $P_0 = P_1 \cdot P_2.$ $P_1$ is bounded below by Proposition~\ref{prop:barrier-safety} and \ref{prop:barrier-reachability}. Assume $\mathcal{A}^{\ast} = \mathcal{A}_i.$ $P_2$ is equivalent to the probability that $\mathcal{U}_{\gamma_i}(t)$ is never eliminated at each time $t,$ which is bounded below by the probability that $\{\sup_{t \in [0,T]} ||K(t)\left(\mathbf{\hat{x}}_{\alpha,i}(t) - \mathbf{\hat{x}}_{\alpha,i,j}(t)\right) ||_2 \leq \gamma_{min}$ $ \forall j \in \{1,\ldots,q\} \setminus \{i\}\}.$ Thus, we obtain
\begin{align}
    P_2 &= Pr(\cap_{j \in \{1,\ldots,q\} \setminus \{i\}} \sup_{t \in [0, T]} \overline{\Omega}_i^{ij} ) \nonumber\\
    &= 1 - Pr(\cup_{j \in \{1,\ldots,q\} \setminus \{i\}} \sup_{t \in [0, T]} \Omega_i^{ij} ) \nonumber\\
    &\geq 1 - \Sigma_{j \in \{1,\ldots,q\} \setminus \{i\}} Pr(\sup_{t \in [0, T]} \Omega_i^{ij} ) \nonumber\\
    &\geq 1 - \Sigma_{j \in \{1,\ldots,q\} \setminus \{i\}} \eta^{i,j}, \label{eq:P_2-lower-bound}\\
    P_0 &\geq P_1 \cdot \left( 1 - \Sigma_{j \in \{1,\ldots,q\} \setminus \{i\}} \eta^{i,j} \right). 
    \label{eq:P_0-lower-bound}
\end{align}
By choosing $\gamma_i, \forall i \in \{1,\ldots,q\}$ and $\gamma_{min}$ properly, we can make $P_0 \geq \max \{1-\epsilon_s, 1-\epsilon_r\} \ \forall i \in \{1,\ldots,q\}.$
\end{proof}

Theorem~\ref{theorem:control-policy-safety-reachability} implies that there exists a lower bound of the probability that the safety and reachability constraints can be satisfied for $\forall \mathcal{A}_i$. 
We observe that this lower bound given in Equation~\eqref{eq:P_2-lower-bound} depends on $\gamma_i \ \forall i \in \{1,\ldots,q\}$, $\gamma_{min}$, and the noise characteristics of the system.
By choosing appropriate $\gamma_i \ \forall i \in \{1,\ldots,q\}$ and $\gamma_{min},$ and constructing corresponding set of feasible control inputs $\mathcal{U}_{\gamma_i}(t)$, we can attempt to make the lower bound be within $[\max \{1-\epsilon_s, 1-\epsilon_r\}, 1]$.

\subsection{Selection of $\gamma_i$ and $\gamma_{min}$}
\label{subsec:selection_of_gamma}
\begin{algorithm}[h]
	\caption{Algorithm for computing the parameters $\{\gamma_i, i \in 1,\ldots,q \}$ and $\gamma_{min}$ that ensures safety and reachability.}
	\label{algo:gamma-selection}
	\begin{algorithmic}[1]
		\Procedure{$\gamma$\_Selection($\epsilon_s,$ $\epsilon_r,$ $iter\_times,$ $\overline{K},$ $\lambda_{i}^{\ast},$ $\lambda_{i,j}^{\ast},$ $i,j \in 1,\ldots,q \}, i \neq j$)}{}
		\State \textbf{Input}: worst case probability of violating safety
		property $\epsilon_s$, worst case probability of violating reachability property $\epsilon_r$, number of iteration steps $iter\_times$, steady-state LQG controller gain $\overline{K}$, maximum eigenvalue of estimate covariance matrix corresponding to each possible compromised sensor set $\lambda_{i}^{\ast}$, maximum eigenvalue of estimate covariance matrix corresponding to each pair of possible compromised sensor sets $\lambda_{i,j}^{\ast}$
		\State \textbf{Output}: radii of feasible control input sets that satisfy safety and reachability properties $\gamma_i$, their minimum $\gamma_{min}$, $i \in \{1,\ldots,q\}$
		\State $\overline{\epsilon} \leftarrow \max \{1-\epsilon_s, 1-\epsilon_r\}$
		\ForEach {$i \in \{1,\ldots,q\}$}
		    \State $\overline{P}_{1,i} \leftarrow 1,$ $\underline{P}_{1,i} \leftarrow ( \overline{P}_{1,i} - \overline{\epsilon})/iter\_times$
		    \State $\gamma_i \leftarrow$ Barrier\_Certificate($\overline{P}_{1,i}$)
		\EndFor
		
        \State Update $\gamma_{min}$ and $i_{min}$.
		
		\ForEach{$i \in \{1,\ldots,q\}$}
		    \State $\eta^{i,j} \leftarrow \frac{4(\lambda_i^{\ast}\Gamma_i + \lambda_{i,j}^{\ast}\Gamma_{i,j})\overline{K}^{2}}{\gamma_{min}^{2}}$
		    \State Update $\overline{P}_{2,i}$ and $\overline{P}_{0,i}$ via~(\ref{eq:P_2-lower-bound}) and~(\ref{eq:P_0-lower-bound}).
		\EndFor
		
		\State $i \leftarrow 1$
        \While{$i \leq iter\_times$}
            
            \State $\overline{P}_{1,i_{min}} \leftarrow \overline{P}_{1,i_{min}} - \underline{P}_{1,i_{min}}$
            \State $\gamma_{i_{min}} \leftarrow $ Barrier\_Certificate($\overline{P}_{1,i_{min}}$)
		    \State Update $\gamma_{min}$ and $i_{min}$.
		    \ForEach{$i \in \{1,\ldots,q\}$}
		        \State $\eta^{i,j} \leftarrow \frac{4(\lambda_i^{\ast}\Gamma_i + \lambda_{i,j}^{\ast}\Gamma_{i,j})\overline{K}^{2}}{\gamma_{min}^{2}}$
		        \State Update $\overline{P}_{2,i}$ and $\overline{P}_{0,i}$ via~(\ref{eq:P_2-lower-bound}) and~(\ref{eq:P_0-lower-bound}).
		    \EndFor
		    
		    \If{$\forall \overline{P}_{0,i} \geq \overline{\epsilon}$}
		        \State break
		    \EndIf
		    
		    \State $i \leftarrow i + 1$
		\EndWhile
            
        \If{$i > iter\_times$}
            \State \Return{null}
        \EndIf
       
        \State \Return{$\{\gamma_i, i \in 1,\ldots,q \}, \gamma_{min}$}
        \EndProcedure
	 \end{algorithmic}
\end{algorithm}

The parameters $\gamma_i \ \forall i \in \{1,\ldots,q\}$ and $\gamma_{min}$ affect both $P_1$ and $P_2.$ In $P_1$, smaller $\gamma_{i} \ \forall i \in \{1,\ldots,q\}$ can make the system be more difficult to be biased. In $P_2$, larger $\gamma_{min}$ is more likely to make the correct constraint be kept. In order to guarantee safety and reachability for all $\mathcal{A}_i$, we need to keep $P_0 \geq \max \{1-\epsilon_s, 1-\epsilon_r\} \ \forall i \in \{1,\ldots,q\}$. Since the closed form expression of $P_1$ as a function of $\gamma_{i}$ is not clear, we utilize a heuristic method to search for proper value of $\gamma_i \ \forall i \in \{1,\ldots,q\}$ and $\gamma_{min}$ which satisfies $P_1 \cdot P_2 \geq \max \{1-\epsilon_s, 1-\epsilon_r\}$ for all $\mathcal{A}_i.$ The intuition is that we initialize $P_1$ first to compute the corresponding $\gamma_{i}$ to the candidate $P_1,$ and calculate the corresponding $\gamma_{min} = \min_{i} \gamma_i$ for all $i \in \{1,\ldots,q\}.$ Then we check the lower bound of $P_0$ for all $\mathcal{A}_i$ according to inequality~(\ref{eq:P_0-lower-bound}). If the lower bound is greater than $\max \{1-\epsilon_s, 1-\epsilon_r\}$ for all $\mathcal{A}_i,$ the safety and reachability constraints are satisfied. Otherwise, we enlarge $\gamma_{min}$ by reducing the $P_1$ corresponding to the minimal $\gamma_i$, recalculating this $\gamma_i$ under the updated $P_1,$ and checking the lower bounds of $P_0$ for all $\mathcal{A}_i$ under new $\gamma_{min}.$ We do this procedure iteratively until the lower bounds of $P_0$ for all $\mathcal{A}_i$ are greater than $\max \{1-\epsilon_s, 1-\epsilon_r\}$ for all $\mathcal{A}_i.$ The proposed procedure is shown in Algorithm~\ref{algo:gamma-selection}.

\section{Control Strategy for Single-adversary Scenario}
\label{Single-adversarial-scenario}


In this section, we consider a special case where there exists a unique attack pattern, denoted as $\mathcal{A}_1.$ Both the controller and the adversary have the knowledge of $\mathcal{A}^{\ast} = \mathcal{A}_1.$ However, the controller does not know which sensors in $\mathcal{A}_1$ are compromised by the adversary. 

In the single-adversary scenario, Equation~\eqref{eq:problem-form-revised multi} in the multiple-adversary scenario becomes
 \begin{align}
 \label{eq:problem-form-revised}
 \min_{\mathbf{u}(t)} \ & \mathbf{E}[\int_{0}^{T}{((\mathbf{x}(t)-\mathbf{r}(t))^{T}Q(\mathbf{x}(t)-\mathbf{r}(t)) }   + \mathbf{u}(t)^{T}R\mathbf{u}(t))  dt \nonumber \\
 &+(\mathbf{x}(T)-\mathbf{r}(T))^{T}F(\mathbf{x}(T)-\mathbf{r}(T)) ] \\
 \mbox{s.t.} \ & \mathbf{u}(t) \in \mathcal{U}_{\gamma_1}(t), \forall t \in [0,T]. \label{eq:problem-form-revised con} 
 \end{align}

When solving~\eqref{eq:HJB multi}, we do not know which $\mathcal{A}_i$ is $\mathcal{A}^{\ast},$ and $\mathcal{I}(t)$ is changing at each time step $t.$ Thus, we have to relax the constraints $\mathcal{U}_{\gamma_i}(t)$ for all $i \in \mathcal{I}(t).$ Since $\mathcal{A}^{\ast} = \mathcal{A}_1$ is known in the single-adversary scenario, we use the method of Lagrange multipliers to construct the dual problem of Equation~(\ref{eq:problem-form-revised}) rather than relax the constraint $\mathcal{U}_{\gamma_1}(t)$ when solving the HJB equation. Denote the Lagrange multiplier corresponding to constraint $\mathcal{U}_{\gamma_1}(t)$ as $\lambda(t)$. 
We form the unconstrained primal problem of Equation~(\ref{eq:problem-form-revised}) as
\begin{align}
\label{eq:problem-form-primal}
\begin{split}
    \min_{\mathbf{u}(t)} \max_{\lambda(t)} \ & \mathbf{E}\{\int_{0}^{T}{[(\mathbf{x}(t)-\mathbf{r}(t))^{T}Q(\mathbf{x}(t)-\mathbf{r}(t)) }   + \mathbf{u}(t)^{T}R\mathbf{u}(t) \\
    &- \lambda(t)\left(\gamma_1^2 - \left(\mathbf{u}(t) - \mathbf{u}_{\alpha,1}(t)\right)^T\left(\mathbf{u}(t) - \mathbf{u}_{\alpha,1}(t)\right)\right) ]  dt \\
    &+(\mathbf{x}(T)-\mathbf{r}(T))^{T}F(\mathbf{x}(T)-\mathbf{r}(T)) \} \\
    \mbox{s.t.} \ &  \lambda(t) \geq 0, \forall t \in [0,T]
\end{split}
\end{align}

The dual problem of~(\ref{eq:problem-form-revised}) is~\cite{shapiro2005duality} 
\begin{subequations}\label{eq:problem-form-dual}
 \begin{align}
 \max_{\lambda(t)} &\min_{\mathbf{u}(t)} \ \mathbf{E}\{\int_{0}^{T}{[(\mathbf{x}(t)-\mathbf{r}(t))^{T}Q(\mathbf{x}(t)-\mathbf{r}(t)) }   + \mathbf{u}(t)^{T}R\mathbf{u}(t) \nonumber \\
 &\qquad- \lambda(t)\left(\gamma_1^2 - \left(\mathbf{u}(t) - \mathbf{u}_{\alpha,1}(t)\right)^T\left(\mathbf{u}(t) - \mathbf{u}_{\alpha,1}(t)\right)\right) ]  dt \nonumber \\
 & \qquad +(\mathbf{x}(T)-\mathbf{r}(T))^{T}F(\mathbf{x}(T)-\mathbf{r}(T)) \} \label{eq:problem-form-dual obj}\\
 \mbox{s.t.} \ &  \lambda(t) \geq 0, \forall t \in [0,T]
 \end{align}
 \end{subequations}
Choosing the value of $\lambda(t)$ at each time step $t$ is challenging, so we relax the problem by assuming $\lambda(t) \equiv \lambda \geq 0$ for all $t \in [0,T]$. The following analysis can be repeated for different values of $\lambda$ to obtain a lower bound on the solution to (20), and hence a lower bound on the value function. The inner minimization problem of Equation~\eqref{eq:problem-form-dual} can be rewritten as

 \begin{align}
 &\min_{\mathbf{u}(t)} \ \mathbf{E}\{\int_{0}^{T}{[(\widetilde{\mathbf{x}}(t) - \widetilde{\mathbf{r}}(t))^T\widetilde{Q}(t)(\widetilde{\mathbf{x}}(t) - \widetilde{\mathbf{r}}(t))} \nonumber \\
 &\qquad+ \mathbf{u}(t)^{T}(\lambda I + R)\mathbf{u}(t)  - \widetilde{\mathbf{x}}(t)^TM(t)\mathbf{u}(t) - \lambda \gamma_1^2 ]  dt \nonumber \\
 & \qquad +(\widetilde{\mathbf{x}}(T)-\widetilde{\mathbf{r}}(T))^{T}\widetilde{F}(\widetilde{\mathbf{x}}(T)-\widetilde{\mathbf{r}}(T)) \}, \label{eq:problem-form-dual obj2}
 \end{align}
where 
$\widetilde{\mathbf{x}}(t) = [\mathbf{x}(t)^T, \hat{\mathbf{x}}_{\alpha,1}(t)^T, \mathbf{s}(t)^T]^T,$ $\widetilde{\mathbf{r}}(t) = [\mathbf{r}(t)^T, 0, 0]^T,$
\begin{equation*}
    \widetilde{Q}(t) = 
    \begin{bmatrix}
        Q & 0 & 0 \\
        0 & \lambda K(t)^TK(t) & -\frac{1}{2}\lambda K(t)^TR^{-1}B^T \\
        0 & -\frac{1}{2}\lambda BR^{-1}K(t) & \frac{1}{4}\lambda BR^{-1}R^{-1}B^T
    \end{bmatrix}
\end{equation*}
\begin{equation*}
    M(t) = 
    \begin{bmatrix}
        0 \\
        2\lambda K(t)^T \\
        -\lambda BR^{-1}
    \end{bmatrix},
    \widetilde{F}(t) = 
    \begin{bmatrix}
        F & 0 & 0 \\
        0 & 0 & 0 \\
        0 & 0 & 0
    \end{bmatrix}
\end{equation*}
The system state $\widetilde{\mathbf{x}}(t)$ evolves as the system dynamics
\begin{equation}
    \dot{\widetilde{\mathbf{x}}}(t) = \widetilde{A}(t)\widetilde{\mathbf{x}}(t) + \widetilde{B}\mathbf{u}(t) + H\widetilde{\mathbf{r}}(t) +  \widetilde{N}\widetilde{\mathbf{w}}(t)
\end{equation}
where
\begin{equation*}
    \widetilde{A}(t) = 
    \begin{bmatrix}
        A & 0 & 0 \\
        \widetilde{\Theta}(t) & A-\widetilde{\Theta}(t) & 0 \\
        0 & 0 & -A^T + \frac{1}{2}P(t)BR^{-1}B^T
    \end{bmatrix}
\end{equation*}
\begin{equation*}
    \widetilde{\Theta}(t) = \Theta_{\alpha,1}(t)C_{\alpha,1},
    \widetilde{B} = 
    \begin{bmatrix}
        B \\
        B \\
        0
    \end{bmatrix},
    H = 
    \begin{bmatrix}
        0 & 0 & 0 \\
        0 & 0 & 0 \\
        2Q & 0 & 0 
    \end{bmatrix}
\end{equation*}
\begin{equation*}
    \widetilde{N} = 
    \begin{bmatrix}
        I & 0 \\
        0 & \Theta_{\alpha,1}(t)
    \end{bmatrix},
    \widetilde{\mathbf{w}}(t) = 
    \begin{bmatrix}
        \mathbf{w}(t) \\
        \mathbf{v}_{\alpha,1}(t)
    \end{bmatrix}
\end{equation*}

The solution of Equation~(\ref{eq:problem-form-dual obj2}) can be obtained by solving a stochastic HJB equation~\cite{kirk2004optimal}. 
\begin{align}
\label{eq:HJB}
    0 = &\min_{\mathbf{u}(t)}{\left\{(\widetilde{\mathbf{x}}(t) - \widetilde{\mathbf{r}}(t))^T\widetilde{Q}(t)(\widetilde{\mathbf{x}}(t) - \widetilde{\mathbf{r}}(t)) \right.} \nonumber\\
    &+ \mathbf{u}(t)^{T}(\lambda I + R)\mathbf{u}(t) 
    - \widetilde{\mathbf{x}}(t)^TM(t)\mathbf{u}(t)\nonumber\\
    &- \lambda \gamma_1^2 + V_{\widetilde{\mathbf{x}}}(t,\widetilde{\mathbf{x}})(\widetilde{A}(t)\widetilde{\mathbf{x}}(t) + \widetilde{B}\mathbf{u}(t) + H\widetilde{\mathbf{r}}(t)) \nonumber\\ 
    &+ \frac{1}{2}\mathbf{tr}(V_{\widetilde{\mathbf{x}}\widetilde{\mathbf{x}}}(t,\widetilde{\mathbf{x}})\widetilde{N}\Sigma_{\widetilde{\mathbf{w}}}\widetilde{N}^T) + \left.V_{t}(t,\widetilde{\mathbf{x}})  \right\}
\end{align}
where $V(t,\widetilde{\mathbf{x}})$ is the value function of equation~(\ref{eq:problem-form-dual obj2}), $V_{\widetilde{\mathbf{x}}}$ and $V_{\widetilde{\mathbf{x}}\widetilde{\mathbf{x}}}(t,\widetilde{\mathbf{x}})$ are the first and second derivatives with respect to $\widetilde{\mathbf{x}},$ and $V_{t}(t,\widetilde{\mathbf{x}})$ is the first derivative with respect to $t.$ This value function is the lower bound of the value function of Equation~(\ref{eq:problem-form-revised}) if the control input $\mathbf{u}(t)$ satisfies Equation~(\ref{eq:problem-form-revised con}) because the term $\lambda\left(\gamma_1^2 - \left(\mathbf{u}(t) - \mathbf{u}_{\alpha,1}(t)\right)^T\left(\mathbf{u}(t) - \mathbf{u}_{\alpha,1}(t)\right)\right)$ is nonnegative. The optimal control input $\mathbf{u}^{\ast}(t)$ of Equation~(\ref{eq:problem-form-dual obj2}) is equal to the minimizer of Equation~(\ref{eq:HJB}) for all $t \in [0, T].$ 

The value function is represented as \cite{anderson2007optimal}
\begin{align}
    \label{eq:value-func1}
    V(t,\widetilde{\mathbf{x}}) &= \frac{1}{2}\widetilde{\mathbf{x}}(t)^{T}\widetilde{P}(t)\widetilde{\mathbf{x}}(t) + \widetilde{\mathbf{x}}(t)^{T}\widetilde{\mathbf{s}}(t) + \widetilde{s}_{0}(t) + \widetilde{\beta}(t) \\
    \label{eq:value-func2}
    V_{\widetilde{\mathbf{x}}}(t,\widetilde{\mathbf{x}}) &= \widetilde{\mathbf{x}}(t)^T\widetilde{P}(t) + \widetilde{\mathbf{s}}(t)^T \\
    \label{eq:value-func3}
    V_{\widetilde{\mathbf{x}}\widetilde{\mathbf{x}}}(t,\widetilde{\mathbf{x}}) &= \widetilde{P}(t) \\
    \label{eq:value-func4}
    V_t(t,\widetilde{\mathbf{x}}) &= \frac{1}{2}\widetilde{\mathbf{x}}(t)^{T}\dot{\widetilde{P}}(t)\widetilde{\mathbf{x}}(t) + \widetilde{\mathbf{x}}(t)^{T}\dot{\widetilde{\mathbf{s}}}(t) + \dot{\widetilde{s}}_{0}(t) + \dot{\widetilde{\beta}}(t) 
\end{align}

Substitute the equations~(\ref{eq:value-func1})-(\ref{eq:value-func4}) into the equation~(\ref{eq:HJB}) and let the value function satisfy the HJB equation for all $\widetilde{\mathbf{x}}(t)$, we have
\begin{align}
    -\dot{\widetilde{P}}(t) &= \widetilde{A}^{T}\widetilde{P}(t) + \widetilde{P}(t)\widetilde{A} + 2\widetilde{Q} \nonumber\\
    - \frac{1}{2}&(\widetilde{B}^{T}\widetilde{P}(t) - M(t)^T)^T(R+\lambda I)^{-1}(\widetilde{B}^{T}\widetilde{P}(t) - M(t)^T)  \\
    \dot{\widetilde{\mathbf{s}}}(t) &= (-\widetilde{A}^{T} + \frac{1}{2}\widetilde{P}(t)\widetilde{B}(R+\lambda I)^{-1}\widetilde{B}^T \nonumber\\
    &- \frac{1}{2}M(t)\widetilde{B})\widetilde{\mathbf{s}}(t) + (2\widetilde{Q} - \widetilde{P}(t)H)\widetilde{\mathbf{r}}(t) \\
    \dot{\widetilde{s}}_{0}(t) &=  \frac{1}{4}\widetilde{\mathbf{s}}(t)^T\widetilde{B}(R+\lambda I)^{-1}\widetilde{B}^T\widetilde{\mathbf{s}}(t) - \widetilde{\mathbf{r}}(t)^{T}Q\widetilde{\mathbf{r}}(t) \nonumber\\
    &- \widetilde{\mathbf{s}}(t)^TH\widetilde{\mathbf{r}}(t) + \lambda \gamma_1^2\\
    -\dot{\widetilde{\beta}}(t) &= \frac{1}{2}\mathbf{tr}(\widetilde{P}(t)\widetilde{N}\Sigma_{\widetilde{\mathbf{w}}}\widetilde{N}^T) 
\end{align}
with boundary conditions $\widetilde{\mathbf{s}}(T) = -2\widetilde{F}\widetilde{\mathbf{r}}(T)$ and $\widetilde{P}(T) = 2\widetilde{F}.$

The minimizer of equation~(\ref{eq:HJB}) for all $t \in [0, T]$ is 
\begin{align}
    \mathbf{u}^{\ast}(t) = &- \frac{1}{2}(R+\lambda I)^{-1}(\widetilde{B}^{T}\widetilde{P}(t) - M(t)^T)\widetilde{\mathbf{x}}(t) \nonumber\\
    &- \frac{1}{2}(R+\lambda I)^{-1}\widetilde{B}^{T}\widetilde{\mathbf{s}}(t).
    \label{eq:controller-single}
\end{align}

Define the optimal value of Equation~\eqref{eq:problem-form-revised} as $\mathcal{V}_1,$ the value of Equation~\eqref{eq:problem-form-revised} using the solution of QCQP~\eqref{eq:approx_opt} as $\mathcal{V}_2,$ and the value of Equation~\eqref{eq:problem-form-revised} using Equation~\eqref{eq:controller-single} as $\mathcal{V}_3.$ Based on the weak duality~\cite{shapiro2005duality},
$\mathcal{V}_3 \leq \mathcal{V}_1.$ Note that when $\lambda=0,$ $\mathcal{V}_3  = \mathcal{V}_2.$ Since $0$ is a feasible solution of $\lambda$ and Equation~\eqref{eq:problem-form-dual obj} maximizes over $\lambda$, we have that $\mathcal{V}_2  \leq \mathcal{V}_3,$
which further yields that $\mathcal{V}_3$ is a tighter bound to $\mathcal{V}_1$ than $\mathcal{V}_2$.

\section{Case study}
\label{sec:simulation}

\begin{figure}
    \centering
    \includegraphics[width=3.5in]{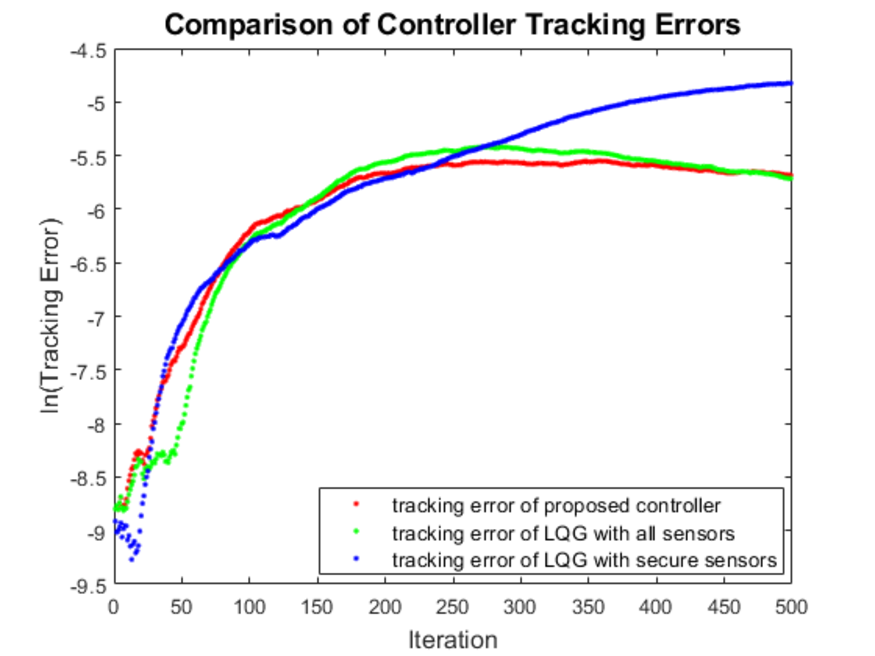}
    \caption{Evaluation of our proposed approach on a linear model case study with no attack. The tracking error of the proposed scheme and the LQG controller using all measurements is less than the tracking error of the LQG controller eliminating measurements indexed by either $\mathcal{A}_1$ or $\mathcal{A}_2$ during the first 200 time steps. }
    \label{fig:no_attack_tracking_error}
\end{figure}

\begin{figure}
    \centering
    \includegraphics[width=3.3in]{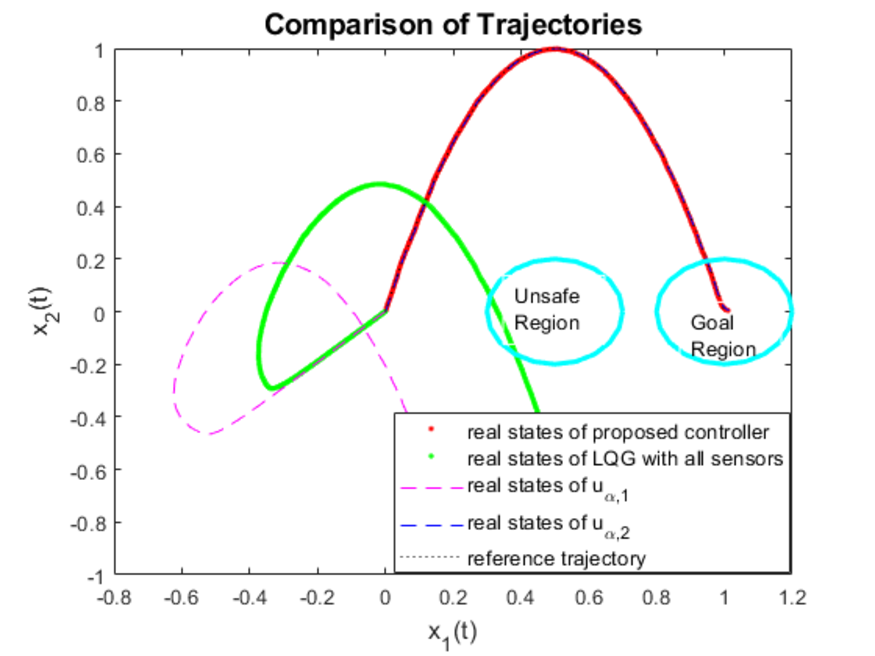}
    \caption{The states of the proposed policy and $\mathbf{u}_{\alpha,2}(t)$ converge to the goal region without reaching the unsafe region in spite of a constant attack. Meanwhile the state of the LQG controller using all measurements and $\mathbf{u}_{\alpha,1}(t)$ violates safety and/or reachability constraints.}
    \label{fig:attack-trajectory}
\end{figure}

\begin{figure}
    \centering
    \includegraphics[width=3.3in]{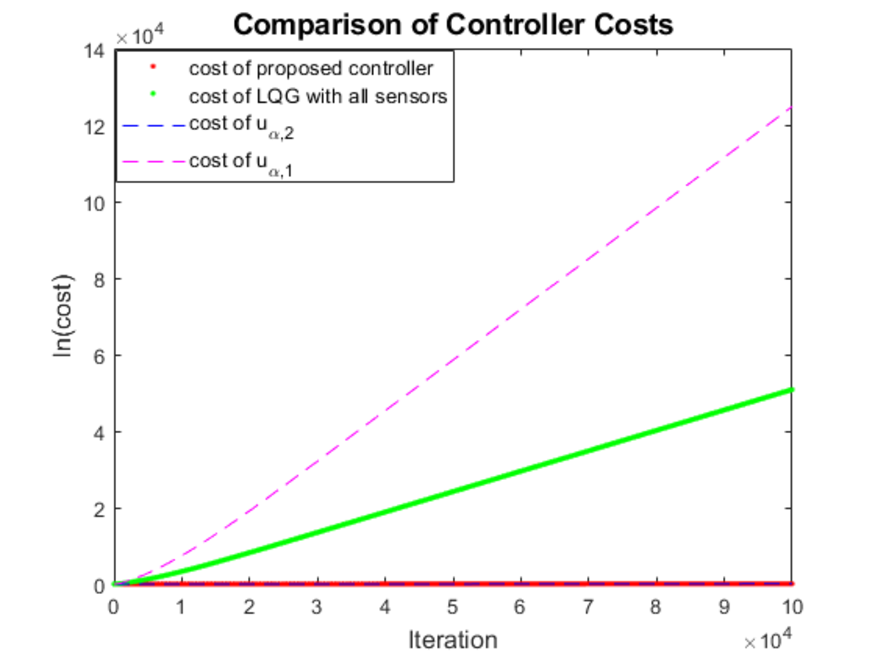}
    \caption{The average costs of the proposed policy and $\mathbf{u}_{\alpha,2}(t)$ are $31.8519$ and $29.6224,$ respectively, which are lower than those of the LQG controllers using all measurements and $\mathbf{u}_{\alpha,1}(t)$.}
    \label{fig:attack-cost}
\end{figure}

In this section, we investigate the proposed scheme in the scenarios where the adversary is present and absent. 


\subsection{System Model}
\label{subsec:system-model}
We consider a system with $2$ states, $2$ inputs, and $4$ sensors. There are $q = 2$ attack patterns, designed as $\mathcal{A}_1 = \{2\}$ and $\mathcal{A}_2 = \{4\}.$ The matrices $A$ and $B$ are set to be identity matrices $I$ with proper dimensions. The matrix $C$ is constructed as 
$C = {\begin{pmatrix}1&1&0&0\\0&0&1&1\end{pmatrix}}^T.$
The noises $\mathbf{w}(t)$ and $\mathbf{v}(t)$ are Gaussian processes with means $0$ and covariances $\Sigma_{\mathbf{w}} = 0.1I$ and $\Sigma_{\mathbf{v}} = 0.1I.$ 
The cost matrices are selected as $Q = I,$ $F = 0.03 I,$ and $R = 1\times10^{-3}I$ with proper dimensions. 

The system tracks a parabolic reference trajectory with the form 
$x_2(t) = -4\times(x_1(t) - 0.5001)^2 + 0.9999,$
where $x_1(t)$ and $x_2(t)$ denote the first and second dimensions of the state at time $t.$  The initial state $\mathbf{x}_0 = (0.0001, -0.0001)^T$. We choose that $\Phi_{\alpha,i}(0) = 10I$ and $\hat{\mathbf{x}}_{\alpha,i}(0) = (0, 0)^T$. The system is required to avoid the set of unsafe states $U = \{\mathbf{x}: 0.2^2 - (x_1 - 0.5001)^2 - (x_2 + 0.0001)^2 > 0\}$ and reach the set of the goal states $G = \{\mathbf{x}: 0.2^2 - (x_1 - 1.0001)^2 - (x_2 + 0.0001)^2 > 0\}$ at $T = 10.$ The worst case probability that the safety and reachability constraints are violated are designed as $\epsilon_s = 0.3$ and $\epsilon_r = 0.3.$

\subsection{Scenario with No Adversary}
In the scenario where the adversary is absent, we set $\mathbf{a}(t) = (0, 0, 0, 0)^T$ and the proposed scheme is compared with two LQG controllers. One LQG controller utilizes the measurements of all sensors, while the other LQG controller eliminates the measurements of sensors indexed by either $\mathcal{A}_1$ or $\mathcal{A}_2$ and utilizes only the measurements of the secure sensors. As shown in Fig.~\ref{fig:no_attack_tracking_error}, all three controllers track the reference trajectory well. The tracking errors of the proposed scheme and the LQG controller using all measurements are less than the tracking error of the LQG controller eliminating measurements indexed by $\mathcal{A}_1$ or $\mathcal{A}_2$ during the early stage. The reason is that at the early stage the gain of the KF $K(t)$ has not converged, so controllers with more sensor measurements can reduce the influence of the noise. Over all time, the average tracking errors of the proposed controller and the LQG controller utilizing measurements of all sensors are $0.0017.$ The average tracking error of the LQG controller eliminating the measurements of sensors indexed by $\mathcal{A}_1$ or $\mathcal{A}_2$ is $0.0021.$ Compared with the LQG controller eliminating the measurements of sensors indexed by $\mathcal{A}_1$ or $\mathcal{A}_2$, the proposed controller and the LQG controller utilizing measurements of all sensors decrease the tracking error for $19\%.$

\subsection{Scenario with Adversary}
In the scenario where the adversary is present, we still consider the system described in Section~\ref{subsec:system-model}. However, there are $p = 6$ sensors and $q = 3$ attack patterns, denoted as $\mathcal{A}_1 = \{1,4\},$ $\mathcal{A}^\ast = \mathcal{A}_2 = \{2,5\},$ and $\mathcal{A}_3 = \{3,6\}.$ The adversary selects $\mathbf{a}(t) = (0, 1, 0, 0, 1, 0)^T.$ Matrix
$C = {\begin{pmatrix}1&1&1&0&0&0\\0&0&0&1&1&1\end{pmatrix}}^T.$

In this scenario, $\mathcal{A}_1 \cup \mathcal{A}_2 \cup \mathcal{A}_3 = \{1,\ldots,p\}$ so there is no secure sensor. The proposed policy is compared with three LQG controllers. One LQG controller utilizes the measurements of all sensors, while the other two LQG controllers $\mathbf{u}_{\alpha,1}(t)$ and  $\mathbf{u}_{\alpha,2}(t)$ which eliminate the measurements of sensors indexed by $\mathcal{A}_1$ and $\mathcal{A}_2$, respectively. The performances of the four controllers are shown in Fig.~\ref{fig:attack-trajectory} and~\ref{fig:attack-cost}. As shown in Fig.~\ref{fig:attack-trajectory}, the proposed controller and $\mathbf{u}_{\alpha,2}(t)$ can satisfy the safety and reachability constraints, while the LQG controller using all measurements and $\mathbf{u}_{\alpha,1}(t)$ are biased by the attack and violate the constraints. From Fig.~\ref{fig:attack-cost}, we see that the cost of proposed controller converges to the cost of $\mathbf{u}_{\alpha,2}(t)$, which is less than the costs of the LQG controller using all measurements and $\mathbf{u}_{\alpha,1}(t)$. From the results shown in Fig.~\ref{fig:attack-trajectory} and~\ref{fig:attack-cost}, the proposed controller guarantees safety and reachability, and at the same time provides comparable cost performance with $\mathbf{u}_{\alpha,2}(t).$ This results from the fact that after the function $\mathcal{I}$\_Selection eliminates the contraints $\mathcal{U}_{\gamma_1}(t)$ and $\mathcal{U}_{\gamma_3}(t),$ the controller is not biased by the adversary. The LQG controller $\mathbf{u}_{\alpha,2}(t)$ is optimal, but the attack pattern $\mathcal{A}^\ast = \mathcal{A}_2$ is not known by the controller a prior. Thus, in real world, $\mathbf{u}_{\alpha,2}(t)$ is not realizable because $\mathcal{A}^\ast = \mathcal{A}_2$ is only known by the adversary.

\section{Conclusion}
\label{sec:conclusion}
This paper considered the LQG tracking problem with safety and reachability constraints and unknown FDI attack. We assumed that the adversary can compromise a subset of sensors. The controller only knows a collection of possible compromised sensor sets, but has no information about which set of sensors is under attack. We computed a control policy by bounding the control input with a collection of quadratic constraints, each of which corresponds to a possible compromised sensor set. We used a barrier certificate based algorithm to constrain the feasible region of the control policy. We proved that the proposed policy satisfies safety and reachability constraints with desired probability. We provided rules to resolve the possible conflicts between the quadratic constraints. We validated the proposed policy with a simulation study.

\section{Appendix}
\label{sec:appendix}
\textit{Proof of Theorem~\ref{theorem:stability_guarantee_EFK_attack_free}: }
Let $\overline{K} = \sup_{t \in [0,T]} ||K(t)||_2,$ $\lambda_i^{\ast} = \sup_{t \in [0,T]}{\{\lambda_{max}(\Sigma_i(t))\}},$ $\lambda_{i,j}^{\ast} = \sup_{t \in [0,T]}{\{\lambda_{max}(\Sigma_{i,j}(t))\}},$
where $\Sigma_i(t)$ and $\Sigma_{i,j}(t)$ are the covariance matrices of $(\mathbf{x}(t)-\hat{\mathbf{x}}_{\alpha,i}(t))$ and $(\mathbf{x}(t)-\hat{\mathbf{x}}_{\alpha,i,j}(t))$, respectively, and $\lambda_{max}(\cdot)$ denotes the maximum eigenvalue of a matrix.

According to triangle inequality, we obtain
\begin{equation*}
\label{eq:stable-EKF}
    \begin{array}{ll}
        &Pr\left(\sup_{t \in [0,T]} ||K(t)(\mathbf{\hat{x}}_{\alpha,i}(t)-\mathbf{\hat{x}}_{\alpha,i,j}(t))||_2 > \gamma_{min} \right) \\
        \leq &Pr\left(\sup_{t \in [0,T]} ||\mathbf{\hat{x}}_{\alpha,i}(t)-\mathbf{x}(t)||_2 \right.\\
        &\left.+ \sup_{t \in [0,T]} ||\mathbf{\hat{x}}_{\alpha,i,j}(t)-\mathbf{x}(t)||_2 > \frac{\gamma_{min}}{\overline{K}} \right)
    \end{array}
\end{equation*}
In order for $\sup_{t \in [0,T]} ||\mathbf{\hat{x}}_{\alpha,i}(t)-\mathbf{x}(t)||_2 + \sup_{t \in [0,T]} ||\mathbf{\hat{x}}_{\alpha,i,j}(t)-\mathbf{x}(t)||_2 > \frac{\gamma_{min}}{\overline{K}}$ to hold, at least one of $\sup_{t \in [0,T]} ||\mathbf{\hat{x}}_{\alpha,i}(t)-\mathbf{x}(t)||_2$ and $\sup_{t \in [0,T]} ||\mathbf{\hat{x}}_{\alpha,i,j}(t)-\mathbf{x}(t)||_2$ should be greater than $\frac{\gamma_{min}}{2\overline{K}}.$ Thus we can consider the right hand side of inequality~(\ref{eq:stable-EKF}) as a union probability, and apply the rule of addition on the union probability
\begin{equation}
    \begin{array}{ll}
        &Pr\left(\sup_{t \in [0,T]} ||\mathbf{\hat{x}}_{\alpha,i}(t)-\mathbf{x}(t)||_2 \right.\\
        &\left.+ \sup_{t \in [0,T]} ||\mathbf{\hat{x}}_{\alpha,i,j}(t)-\mathbf{x}(t)||_2 > \frac{\gamma_{min}}{\overline{K}} \right) \\
        \leq &Pr\left(\sup_{t \in [0,T]} ||\mathbf{\hat{x}}_{\alpha,i}(t)-\mathbf{x}(t)||_2 > \frac{\gamma_{min}}{2\overline{K}}\right) \\
        &+ Pr\left(\sup_{t \in [0,T]} ||\mathbf{\hat{x}}_{\alpha,i,j}(t)-\mathbf{x}(t)||_2 > \frac{\gamma_{min}}{2\overline{K}}\right)
    \end{array}
\end{equation}

Define $\mathbf{e}_i(t) = \mathbf{\hat{x}}_{\alpha,i}(t)-\mathbf{x}(t)$ for $\forall i,$ and $\mathbf{e}_{ij}(t) = \mathbf{\hat{x}}_{\alpha,i,j}(t)-\mathbf{x}(t)$ for $\forall i, j.$ We have
\begin{equation*}
    \begin{array}{ll}
        &Pr\left(\sup_{t \in [0,T]} ||\mathbf{\hat{x}}_{\alpha,i}(t)-\mathbf{x}(t)||_2 > \frac{\gamma_{min}}{2\overline{K}}\right) \\
        &+ Pr\left(\sup_{t \in [0,T]} ||\mathbf{\hat{x}}_{\alpha,i,j}(t)-\mathbf{x}(t)||_2 > \frac{\gamma_{min}}{2\overline{K}}\right) \\
        = &Pr\left(\sup_{t \in [0,T]} \mathbf{e}_i(t)^T\mathbf{e}_i(t) > \frac{\gamma_{min}^2}{4\overline{K}^2}\right) \\
        &+ Pr\left(\sup_{t \in [0,T]} \mathbf{e}_{ij}(t)^T\mathbf{e}_{ij}(t) > \frac{\gamma_{min}^2}{4\overline{K}^2}\right) \\
        \leq &Pr\left(\sup_{t \in [0,T]} \mathbf{e}_i(t)^T{\Sigma_i(t)}^{-1}\mathbf{e}_i(t) > \frac{\gamma_{min}^2}{4\overline{K}^2\lambda_i^{\ast}}\right) \\
        &+ Pr\left(\sup_{t \in [0,T]} \mathbf{e}_{ij}(t)^T{\Sigma_{i,j}(t)}^{-1}\mathbf{e}_{ij}(t) > \frac{\gamma_{min}^2}{4\overline{K}^2\lambda_{i,j}^{\ast}}\right) 
    \end{array}
\end{equation*}

For observable systems, the functions $V(\mathbf{e}_i(t),t) = \mathbf{e}_i(t)^{T}\Sigma_i(t)^{-1}\mathbf{e}_i(t)$ and $V(\mathbf{e}_{ij}(t),t) = \mathbf{e}_{ij}(t)^{T}\Sigma_{ij}(t)^{-1}\mathbf{e}_{ij}(t)$ are known to have differential generators that are strictly decreasing \cite{reif2000stochastic}, and hence are supermartingales. Lemma~\ref{lemma:doob_martingale} then implies that 
\begin{equation*}
    \begin{array}{ll}
        &Pr(\sup_{t \in [0,T]} \mathbf{e}_i(t)^T{\Sigma_i(t)}^{-1}\mathbf{e}_i(t) > \frac{\gamma_{min}^2}{4\overline{K}^2\lambda_i^{\ast}})  \\
        &+ Pr(\sup_{t \in [0,T]} \mathbf{e}_{ij}(t)^T{\Sigma_{i,j}(t)}^{-1}\mathbf{e}_{ij}(t) > \frac{\gamma_{min}^2}{4\overline{K}^2\lambda_{i,j}^{\ast}})  \\
        \leq &\left(\frac{\gamma_{min}^{2}}{4\overline{K}^{2}\lambda_i^{\ast}}\right)^{-1}\lim_{t \rightarrow 0}{\left[\mathbf{E}\left(\mathbf{e}_i(t)^T{\Sigma_i(t)}^{-1}\mathbf{e}_i(t)\right)\right]}  \\
        &+ \left(\frac{\gamma_{min}^{2}}{4\overline{K}^{2}\lambda_{i,j}^{\ast}}\right)^{-1}\lim_{t \rightarrow 0}{\left[\mathbf{E}\left(\mathbf{e}_{ij}(t)^T{\Sigma_{ij}(t)}^{-1}\mathbf{e}_{ij}(t)\right)\right]} \\
        = &\frac{4\lambda_i^{\ast}\overline{K}^{2}}{\gamma_{min}^{2}}{\left[\mathbf{E}\left(\mathbf{e}_i(0)^T{\Sigma_i(0)}^{-1}\mathbf{e}_i(0)\right)\right]} \\
        &+ \frac{4\lambda_{i,j}^{\ast}\overline{K}^{2}}{\gamma_{min}^{2}}{\left[\mathbf{E}\left(\mathbf{e}_{ij}(0)^T{\Sigma_{ij}(0)}^{-1}\mathbf{e}_{ij}(0)\right)\right]}.
    \end{array}
\end{equation*}
Letting $\eta^{i,j} = {4(\lambda_i^{\ast}\Gamma_i + \lambda_{i,j}^{\ast}\Gamma_{i,j})\overline{K}^{2}}/{\gamma_{min}^{2}},$ where $\Gamma_i = \mathbf{E}\left(\mathbf{e}_i(0)^T{\Sigma_i(0)}^{-1}\mathbf{e}_i(0)\right)$, $\Gamma_{i,j} = \mathbf{E}\left(\mathbf{e}_{ij}(0)^T{\Sigma_{ij}(0)}^{-1}\mathbf{e}_{ij}(0)\right),$ $\mathbf{e}_i(0) = \hat{\mathbf{x}}_{\alpha,i}(0) - \mathbf{x}_0,$ and $\mathbf{e}_{ij}(0) = \hat{\mathbf{x}}_{\alpha,i,j}(0) - \mathbf{x}_0.$ we have 
\begin{equation*}
    Pr\left(\sup_{t \in [0,T]} ||K(t)\left(\mathbf{\hat{x}}_{\alpha,i}(t) - \mathbf{\hat{x}}_{\alpha,i,j}(t)\right) ||_2 > \gamma_{min} \right) \leq \eta^{i,j}.
\end{equation*}
\begin{flushright}
$\square$
\end{flushright}

\end{document}